\documentclass[journal,comsoc]{IEEEtran}
\usepackage[T1]{fontenc}

\usepackage{url}

\usepackage{subfigure}

\usepackage{times}
\usepackage{amsthm}
\usepackage{amsmath}
\usepackage{amssymb}

\usepackage[pdftex]{graphicx}
\usepackage{wrapfig}
\usepackage{epstopdf}
\usepackage[american]{babel}
\usepackage{url}
\usepackage{color}
\usepackage{xspace}
\usepackage{float}
\usepackage{tabularx}
\usepackage{multirow}
\usepackage{alltt}
\usepackage{multicol}
\usepackage{blindtext}
\usepackage{scrextend}

\usepackage{booktabs}
\usepackage{pifont}
\usepackage{makecell}
\usepackage[super]{nth}

\newtheorem{theorem}{Theorem}

\newcommand{\xmark}{\ding{55}}

\ifCLASSINFOpdf
\else
\fi
\usepackage{amsmath}
\usepackage[cmintegrals]{newtxmath}
\begin{document}
%
\title{Self-healing Dilemmas in Distributed Systems: Fault Correction vs. Fault Tolerance}
%
%
%

\author{Jovan~Nikoli\'c, Nursultan~Jubatyrov and~Evangelos~Pournaras
\thanks{J. Nikoli\'c is with Google, Zurich, Switzerland, e-mail: jovannikolic@google.com.}
\thanks{N. Jubatyrov is with Facebook, London, UK, e-mail: nurs@fb.com.}
\thanks{E. Pournaras is with University of Leeds, Leeds, UK, e-mail: e.pournaras@leeds.ac.uk.}
\thanks{Manuscript received November 20, 2020; revised March, 2021.}}

%
%

\markboth{Journal of \LaTeX\ Class Files,~Vol.~X, No.~X, April~2021}%
{Shell \MakeLowercase{\textit{et al.}}: Bare Demo of IEEEtran.cls for IEEE Communications Society Journals}
%



\maketitle

\begin{abstract}
Large-scale decentralized systems of autonomous agents interacting via asynchronous communication often experience the following self-healing dilemma: fault detection inherits network uncertainties making a remote faulty process indistinguishable from a slow process. In the case of a slow process without fault, fault correction is undesirable as it can trigger new faults that could be prevented with fault tolerance that is a more proactive system maintenance.  But in the case of an actual faulty process, fault tolerance alone without eventually correcting persistent faults can make systems underperforming. Measuring, understanding and resolving such self-healing dilemmas is a timely challenge and critical requirement given the rise of distributed ledgers, edge computing, the Internet of Things in several energy, transport and health applications. This paper contributes a novel and general-purpose modeling of fault scenarios during system runtime. They are used to accurately measure and predict inconsistencies generated by the undesirable outcomes of fault correction and fault tolerance as the means to improve self-healing of large-scale decentralized systems at the design phase. A rigorous experimental methodology is designed that evaluates 696 experimental settings of different fault scales, fault profiles and fault detection thresholds in a prototyped decentralized network of 3000 nodes. Almost 9 million measurements of inconsistencies were collected in a network,  where each node monitors the health status of another node, while both can defect. The prediction performance of the modeled fault scenarios is validated in a challenging application scenario of decentralized and dynamic in-network data aggregation using real-world data from a Smart Grid pilot project. Findings confirm the origin of inconsistencies at design phase and provide new insights how to tune self-healing at an early stage. Strikingly, the aggregation accuracy is well predicted as shown by high correlations and low root mean square errors.
\end{abstract}


%
\IEEEpeerreviewmaketitle

%
%
%
%

\section{Introduction}\label{sec:introduction}

\IEEEPARstart{S}{everal} complex systems in nature and society often exhibit striking reliability, a result of timely choosing, applying and orchestrating multiple self-healing and adaptation strategies. For instance, effectively mitigating blackouts in power grids requires several tailored fault-correction and fault-tolerance mechanisms, whose coordination is way more sophisticated than simply repairing the originating fault of a power line. These include resilient topological design, load-shedding, operating reserves, islanding and active devices among others~\cite{Mei2011}. While a level of sophisticated self-healing in natural systems is usually a result of self-adaptation and evolution, in artificial socio-technical systems with central control such as power grids, reliability remains to a high extent a result of planning based on past experience, adaptations based on precomputed simulations and manual human interventions by system operators. 

Decentralized autonomous systems recently witness a phenomenal rise with the applicability of distributed ledgers, edge computing, multi-agent systems and the Internet of Things in several sectors of society, e.g. energy, transport, health, agriculture, etc~\cite{Xiong2018}. Large-scale asynchronous distributed environments experience unprecedented network/system uncertainties that challenge the orchestration of self-healing strategies: Fault detection inherits these uncertainties that can make a faulty remote process indistinguishable from a slow process~\cite{VanRenesse1998,Lavinia2011}. As a result, a reactive system recovery may turn into an undesirable fault correction of a non-faulty system, which in turn may cause an actual fault that could be prevented instead with fault tolerance, i.e. a more proactive and preventive system maintenance. Figure~\ref{fig:model} illustrates the self-healing dilemma problem studied in this paper.

\begin{figure}[!htb]
\centering
\includegraphics[width=1.0\columnwidth]{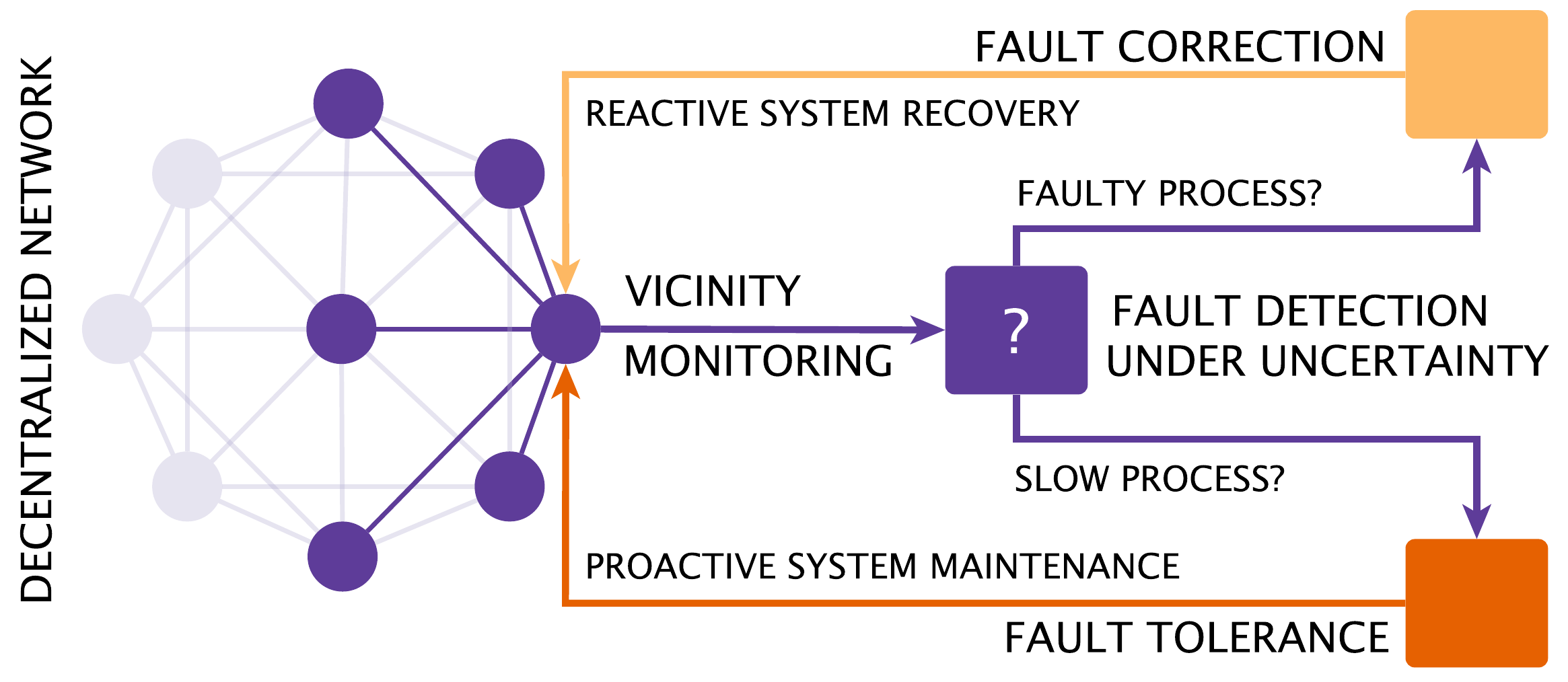}
\caption{The self-healing dilemma of agents comprising a large-scale decentralized networked systems: Fault detection within the vicinity of an agent comes with network uncertainties: Is a remote process faulty or slow? Which self-healing strategy should be applied? Fault correction as a reactive system recovery with the risk of introducing new faults? Or, fault tolerance as a proactive system maintenance with the risk of letting a fault affecting system performance?}\label{fig:model}
\end{figure}

The effective resolution of such fault-correction vs. fault-tolerance dilemmas promises more effective self-healing mechanisms for large-scale decentralized networked systems with uncertainties. A more informed and timely application of fault correction and fault tolerance has the potential to decrease the likelihood of new faults by self-healing itself, against which decentralized systems remain more resilient with a lower communication and processing cost. This paper models and classifies the possible outcomes of self-healing dilemmas between pairs of agents, where one monitors the health status of the other, while both can arbitrary defect. These outcomes are possible desirable and undesirable states (inconsistencies) into which self-healing can fall. The cost of these inconsistencies by undesirable outcomes is further formalized and distinguished within fault scenarios during system runtime. These fault scenarios have the novelty of predicting the performance of self-healing mechanisms without knowledge of the computational/application scenario, overlying algorithms or application data. The modeled fault scenarios are applied and studied to the computational case study of decentralized dynamic in-network aggregation~\cite{Pournaras2017} by introducing a new prototyped fault-detection mechanism based on gossip-based communication~\cite{Shah2009} and agent migrations~\cite{Oyediran2016}. Fault correction and fault tolerance are employed to improve the estimates of aggregation functions made by each node in the network. These estimations approximate, for instance, the total power demand based on which decentralized demand-response programs and power markets operate~\cite{Croce2016}. To preserve accuracy, fault correction performs risky recomputations of the total power demand reactively when faults are detected with uncertainty, while fault tolerance relies entirely on occasional proactive recomputations to capture changes.

A rigorous experimental methodology is introduced to tackle three objectives: (i) Profiling of the inconsistency cost generated by the modeled fault scenarios across 696 experimental settings with varying fault scales, fault profiles and fault-detection thresholds. (ii) Validation of whether the inconsistency cost measured by the modeled fault scenarios is a good general predictor of the accuracy observed in the application scenario of decentralized aggregation of real-world power consumption data. (iii) Comparison of different model calibrations for the prediction of aggregation accuracy by relying on application-independent features.

The findings of the experimental evaluation have significant implications and impact for system designers and operators: By (re)using the general-purpose fault scenarios for vulnerability analysis, the resilience of different system designs can be assessed at an early stage with low cost and under different fault characteristics, while fault-detection mechanisms can be parameterized more effectively. Application developers can improve the self-healing capabilities of applications at the design phase by predicting the impact of faults and tuning appropriately the application before deployment to lower its cost. They can also plan computational resources for self-healing more effectively. 

The contributions of this paper are summarized as follows: (i) The modeling of possible outcomes in agents' self-healing dilemmas. (ii) The modeling of application-independent fault scenarios during system runtime that sufficiently formalize the overall heath status of decentralized systems and their impact on self-healing performance. (iii) A general-purpose novel fault-detection mechanism based on gossip-based communication and migrating agents. (iv) The applicability of the fault-detection mechanism to the Dynamic Intelligent Aggregation Service (DIAS)~\cite{Pournaras2017} for the improvement of its aggregation accuracy. Self-corrective operations of DIAS are expanded when nodes massively fail~\cite{Pournaras2017c}. (v) The profiling of the predicted inconsistencies that different fault scenarios cause under different fault scales, profiles and fault-detection thresholds. (vi) Three model calibration methods to improve the accuracy of the predicted inconsistencies that rely entirely on application-independent features. 

This paper is organized as follows: Section~\ref{sec:related-work} positions and compares this study with related work. Section~\ref{sec:dimemmas} models the uncertainties in fault detection for large-scale asynchronous decentralized systems and introduces the possible outcomes in agents' self-healing dilemmas. Section~\ref{sec:fault-scenarios} formalizes fault scenarios that predict the cost of inconsistencies caused by faults. Section~\ref{sec:model-applicability} illustrates the applicability of the proposed model in a case study of decentralized aggregation. The mechanisms for fault detection, fault correction and fault tolerance to improve aggregation accuracy are outlined. Section~\ref{sec:methodology} introduces the experimental methodology that addresses the objectives of this study and Section~\ref{sec:results} illustrates the findings of the experimental evaluation. Finally, Section~\ref{sec:conclusion} concludes this paper and outlines future work. 	

\section{Positioning and Comparison to Related Work}\label{sec:related-work}

Self-healing mechanisms usually address different types of faults classified according to recent taxonomies~\cite{Stankovic2017,Mukwevho2018}. Assuming reliable communication channels, faults are differentiated as follows~\cite{Stankovic2017}: (i) \emph{Crash} - agents stop responding and terminate. (ii) \emph{Omission} - agents sporadically skip sending/receiving messages. (iii) \emph{Timing} - agents do not complete a task in a certain time frame. (iv) \emph{Arbitrary} (Byzantine)- agents deviate from the expected behavior and operate unpredictably. Another classification distinguishes between (i) \emph{transient}, (ii) \emph{intermittent}, (iii) \emph{permanent} and (iv) \emph{Byzantine} faults~\cite{Mukwevho2018}. While transient faults draw parallels with omission ones and are a result of a temporary affecting condition, e.g. network connectivity, intermittent faults are random, temporary and usually result of hardware failure. In contrast, crash hardware faults require repair of the root cause and are a subset of permanent faults. Byzantine faults result in corrupted/malicious agents sharing manipulated, forged or incorrect data. In large-scale decentralized systems, designing self-healing mechanisms exclusively for certain fault types is ineffective. Several such faults can co-occur, cascade or even have a cause-effect relationship resulting in vicious adaptation cycles, e.g. a faulty fault detection causing faulty fault correction and vice versa. Modeling the interplay of faults and formalizing complex fault scenarios as well as their impact on self-healing performance is fundamental and missing. 

In such fault scenarios, the reliability of fault detection plays a key role~\cite{Stankovic2017}. Two main fault-detection approaches of periodic \emph{heartbeat messages}  and \emph{agent interactions} are identified. The latter further distinguishes between \emph{timeout} and \emph{missing callback} detection. Replication~\cite{Isong2013,Stankovic2017} is a common approach that supports both fault tolerance and fault correction in terms of guarantying the availability of (backup) resources and repair modules for self-healing~\cite{Fedoruk2002}. Replication can be \emph{active vs. passive} based on whether replicas are used only when faults occur~\cite{Marin2001}, \emph{adaptive} based on criteria for replication~\cite{Arfat2016,Almeida2006}, \emph{dynamic} by switching on-the-fly replication schemes~\cite{Marin2001}, or \emph{homogeneous vs. heterogeneous} based on whether replicas are identical copies or equivalent processes~\cite{Fedoruk2002}. Replication is applied to check-point schemes based on rollback protocols~\cite{Jin2004,Koldehofe2013}, consensus protocols and hybrid approaches~\cite{Park2004}. Replication methods are particularly applicable in multi-agent systems, for instance, group replication via proxy servers~\cite{Fedoruk2002}, replication of agents based on the criticality of their planned actions~\cite{Almeida2006}, adjustable group replication with a leader agents~\cite{Marin2001} or introducing a special class of agents for redundancy maintenance~\cite{Kumar2000}. New replication strategies designed for the Internet of Things and cyber-physical systems are subject of recent work~\cite{Terry2016,Ratasich2019}. Self-healing methods~\cite{Sterritt2005} can be \emph{preventive (proactive)}~\cite{Isong2013} or \emph{reactive (resilient)}~\cite{Mukwevho2018}. The former methods require prediction based on probabilistic modeling and monitoring~\cite{Panda2015}. The latter ones require learning capabilities from historic data and observations~\cite{Rahnama2017,Su2016}. 

Despite the large body of work on fault correction and fault tolerance, a recent comprehensive review of such approaches for multi-agent systems identifies as imperative the need for generalized and standardized evaluation of fault-tolerance approaches~\cite{Stankovic2017}. Another recent systematic evaluation of 36 state-of-the-art self-healing systems from the research communities of autonomic computing, self-adaptive, self-organizing and self-managing systems (ICAC, SASO, TAAS, SEAMS) is illustrated~\cite{Ghahremani2020}. Empirical assessments conclude that multiple input traces covering a vast spectrum of failure characteristics are required to predict the performance of a self-healing system. Therefore, generalized models that predict the impact of faults and their correction on large-scale decentralized systems are missing so far~\cite{Isong2013}. Predicting without knowledge about the computational/application scenario, executed application algorithms and application data is challenging~\cite{Isong2013}. Such models have the potential to fundamentally influence the understanding of how to design and deploy more cost-effective decentralized self-healing systems. What makes particularly challenging the inception of such general models is the absence of central control units, the agents' autonomy, the network uncertainties and the non-determinism of system operations~\cite{Fedoruk2002}. In particular, failed processes are often indistinguishable from slow processes in asynchronous decentralized environments that inherit the impossibility of distributed consensus~\cite{Fischer1985,Chandra1996}. As a consequence, fault detection inherits such uncertainties~\cite{VanRenesse1998,Sridhar2006} (is the process faulty or slow?), which in turn results in dilemmas on what self-healing adaptations to apply, i.e. fault correction vs. fault tolerance. This paper addresses these self-healing dilemmas.

\section{Self-healing Dilemmas}\label{sec:dimemmas}

This paper studies self-healing of large-scale decentralized networked systems with faulty nodes. Decentralization means that no single node has full information about all other nodes in the network at a time and each node is connected with a limited number of other nodes. Faults can be a result of system failure, software failure, security attack or any other type of error that makes a faulty node inaccessible to other healthy nodes~\cite{Fedoruk2002}. Nodes usually depend on each other to perform distributed operations by communicating with each other in a peer-to-peer fashion. Even if communication is asynchronous, a fault introduces a cost that hinders (i) performance and/or (ii) consistency of a distributed operation.  The latter is referred to as \emph{inconsistency cost} and is the main focus of this paper. 

Two approaches are distinguished to eliminate these costs: (i) \emph{fault correction} vs. (ii) \emph{fault tolerance}. Fault correction eliminates the performance and inconsistency cost of faults via an effective and timely fault detection and its correction. For instance, consider a master-slave heartbeat mechanism with which a master node monitors the health status of a slave node by receiving periodically heartbeat messages. A fault detection by the master node is the passage of time period without receipt of a heartbeat message. This period is usually selected empirically and universally~\cite{Gyamfi2019}. In contrast, fault-tolerance mechanisms are designed to decrease the performance and inconsistency cost of faults by preventing a total system break down and allowing a system to continue its operation with an operating quality proportional to the severity of the fault. 

The following assumptions are made: (i) Both fault correction and fault tolerance have both, a performance and inconsistency cost. They have performance cost because their operations usually introduce communication and processing overhead. They have inconsistency cost because of uncertainties in fault detection. A fault may be erroneously detected because of high network latency, low convergence speed of the underlying communication model, misconfiguration or poor design in fault detection. For instance, a heartbeat message may not be received because of network fault rather than because of a node fault. The unnecessary correction process consumes resources and introduces potential inconsistencies as system operations are usually not designed to tolerate unnecessary fault corrections. (ii) The performance and inconsistency cost of fault correction and fault tolerance is significantly lower than the respective costs of a system left to be faulty, i.e. without any self-healing. In other words, it makes sense to take care of faults either via fault tolerance or fault correction (or both). (iii) The performance and inconsistency cost of fault correction vs. the ones of fault tolerance depend on the operational state of the system during system runtime and therefore, it is unclear which self-healing approach should be adopted. Based on these three assumptions as well as the focus of this paper on inconsistency cost to eliminate the number of studied dimensions, a study on how to minimize the inconsistency cost in fault correction vs. fault-tolerance dilemmas is illustrated.

Figure~\ref{fig:concept}a models self-healing dilemmas in a decentralized networked system that consists of Node $A$, $B$ and $C$. Each node runs a self-healing agent that is responsible to perform fault correction or fault tolerance. Given the focus of this paper on faulty nodes and without loss of generality, the self-healing agents need to operate remotely so that they are not affected by the faults of the parent nodes, i.e. the original host nodes that created them. As a result, they migrate to neighboring nodes as shown in Figure~\ref{fig:concept}b. In practice, the scope of this model covers several systems that have backup components for redundancy. For the sake of this illustration, a heartbeat mechanism is assumed with which self-healing agents monitor the health status of parent nodes as shown in Figure~\ref{fig:concept}c. Heartbeat messages may not arrive at Node $B$ because of (i) a fault in the parent Node $A$ and/or (ii) a large latency or network error in the link between Node $A$ and $B$~\cite{Lavinia2011}. See Figure~\ref{fig:concept}d. Therefore, when heartbeat messages are not anymore received in Node $A$, the dilemma of the self-healing agent is whether to perform fault correction, i.e. establish the new link between $A$ and $C$ because the parent Node $B$ is truly faulty or perform fault tolerance, i.e. not establish a new link because missing heartbeat messages are probably a result of high latency or fault in the link connecting Node $A$ and $B$. It is assumed that the applied fault correction is effective if and only if nodes experience faults, otherwise correction introduces an inconsistency cost.  

\begin{figure}[!htb]
	\centering
	\includegraphics[width=0.6\columnwidth]{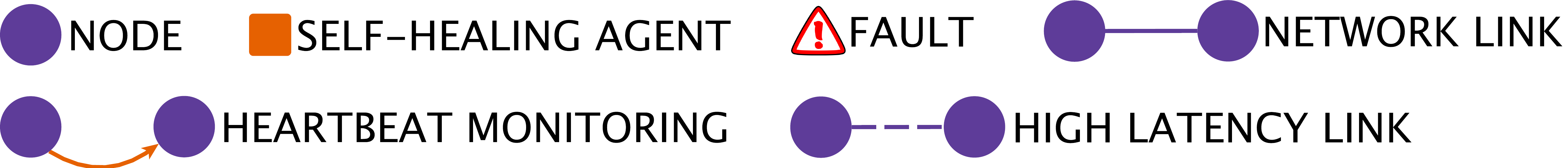}\\
	\subfigure[]{\includegraphics[width=0.24\columnwidth]{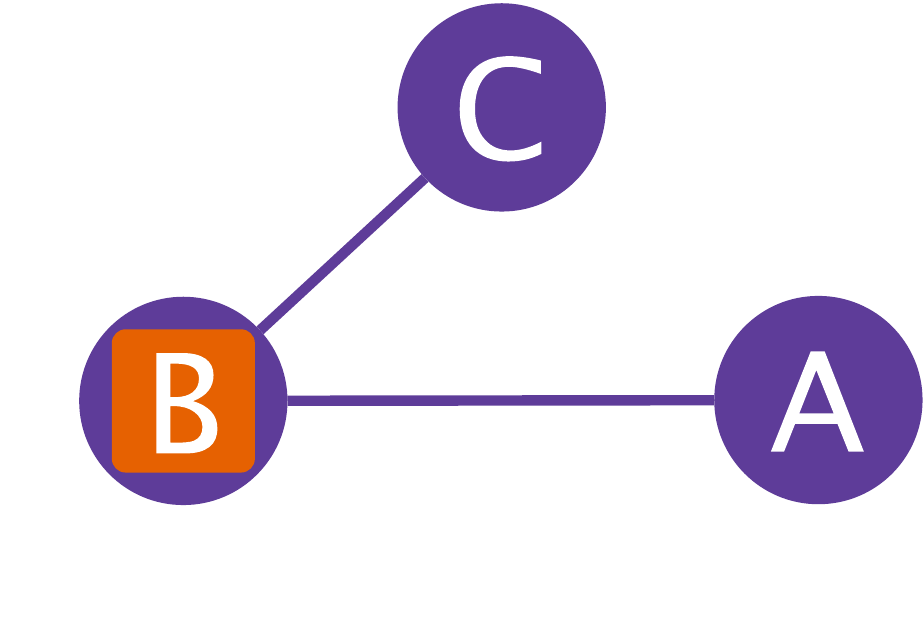}}
	\subfigure[]{\includegraphics[width=0.24\columnwidth]{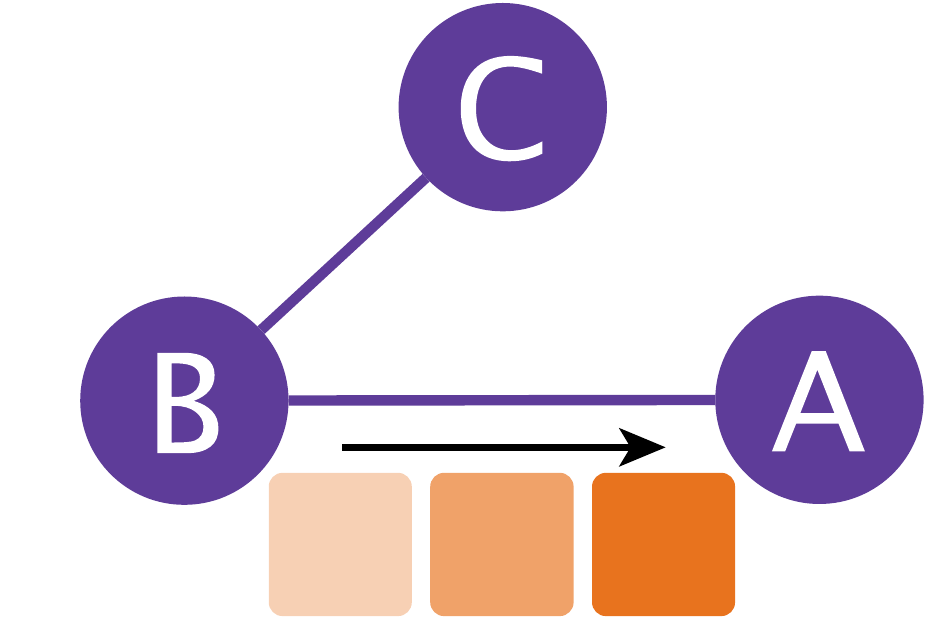}}
	\subfigure[]{\includegraphics[width=0.24\columnwidth]{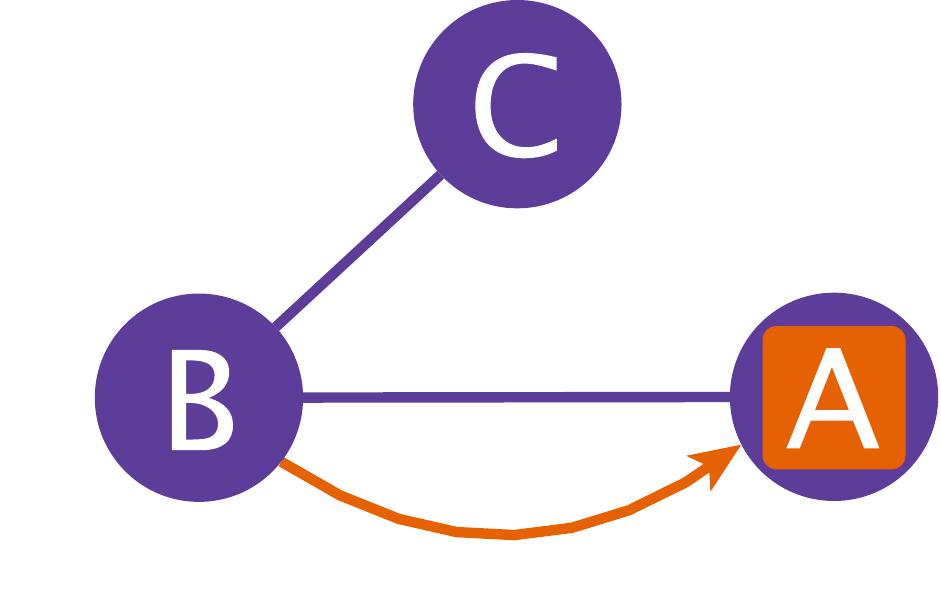}}
	\subfigure[]{\includegraphics[width=0.24\columnwidth]{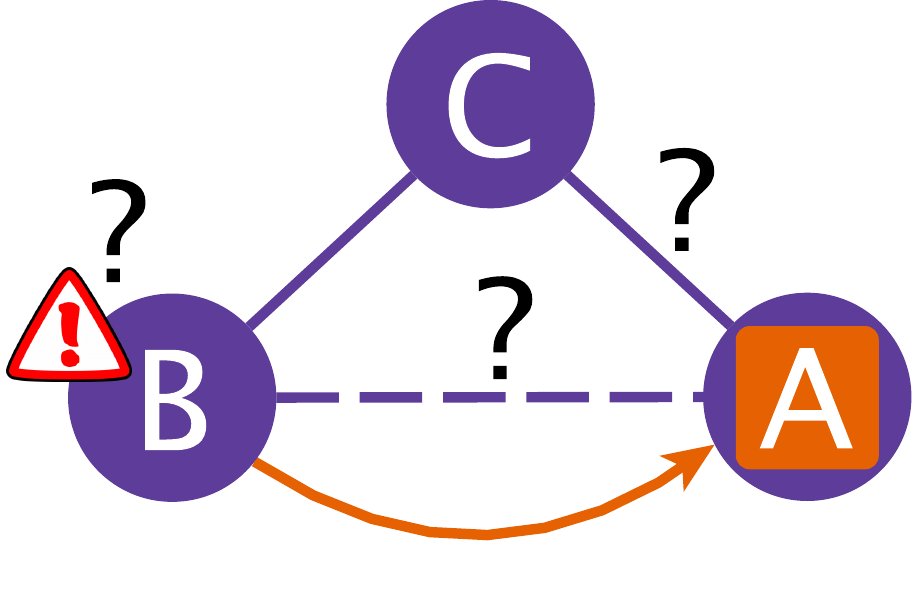}}
	\caption{Modeling self-healing dilemmas in a decentralized networked system. \textbf{(a)} Nodes (circles) are connected in a decentralized network and initiate self-healing agents (squares), which are responsible to detect, correct or tolerate faults of their parent node. \textbf{(b)} Self-healing agents find a host node to migrate for redundancy. In this way, a fault on their parent node does not influence their self-healing operations. \textbf{(c)} Self-healing agents monitor the health status of their parent nodes via, for instance, heartbeat signals. \textbf{(d)} Fault detection, determined by a time period during which heartbeat messages are not received from the parent node, comes along with uncertainties. For instance, the heartbeat messages may not be received because of high latency or network error on the $A$-$B$ link, rather than because Node $B$ is faulty~\cite{Lavinia2011}. In this example, self-healing agents perform fault correction by initiating a new connection with another node if they detect a fault in the parent node. Fault correction eliminates inconsistency cost if a node is actually faulty, but introduces inconsistency cost if the parent node is actually not faulty.   Therefore, the self-healing agent has the following dilemma: Should it establish the $A$-$C$ link (fault correction) or wait longer for heartbeat messages to arrive (fault tolerance)? There are four possible outcomes in this decision-making illustrated in Figure~\ref{fig:outcomes}.}\label{fig:concept}
\end{figure}

The fault-correction vs. fault-tolerance dilemmas come with four possible outcomes as shown in Figure~\ref{fig:outcomes}. Note that in Figure~\ref{fig:outcomes}a and~\ref{fig:outcomes}b there are two outcomes that do not have inconsistency cost (desirable outcomes). These are the \emph{true negative} outcome that is a result of effective fault tolerance and the \emph{true positive} outcome as a result of effective fault correction. Figure~\ref{fig:outcomes}c and~\ref{fig:outcomes}d show the two outcomes with an inconsistency cost (undesirable outcomes). These are the \emph{false negative} outcome\footnote{False negatives also originate from faults in the node hosting the self-healing agent.} by erroneous fault tolerance and the \emph{false positive} outcome by erroneous fault correction. 

\begin{figure}[!htb]
	\centering
	\includegraphics[width=0.6\columnwidth]{legend.pdf}\\
	\subfigure[True negative. Desirable.]{\includegraphics[width=0.24\columnwidth]{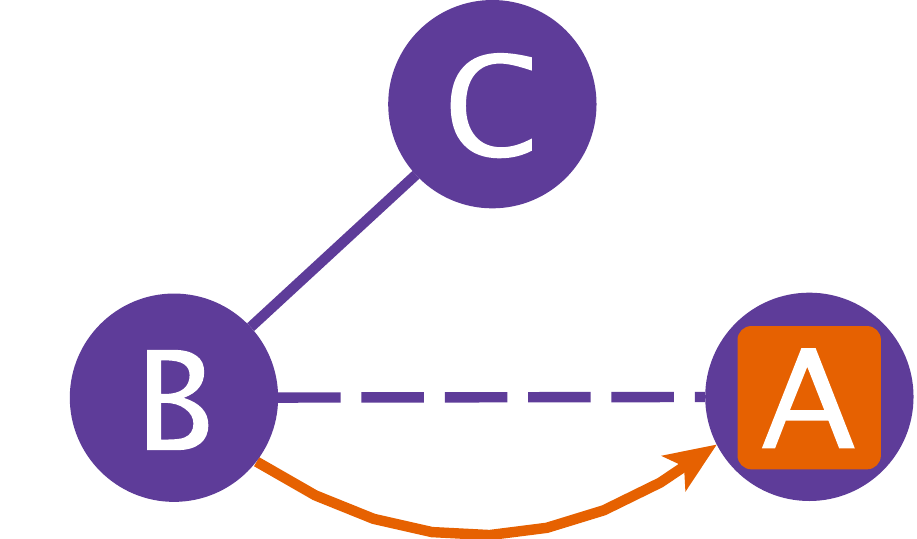}}
	\subfigure[True positive. Desirable.]{\includegraphics[width=0.24\columnwidth]{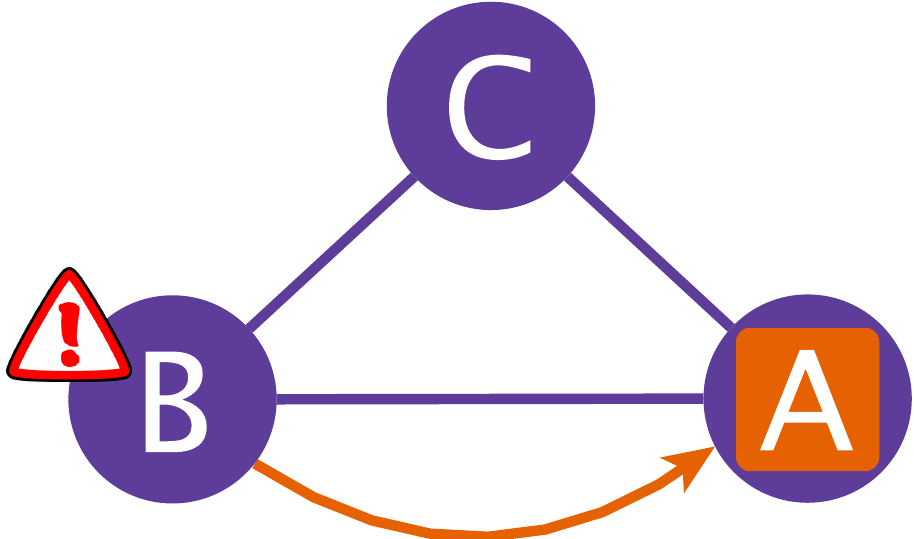}}
	\subfigure[False negative. Undesirable.]{\includegraphics[width=0.24\columnwidth]{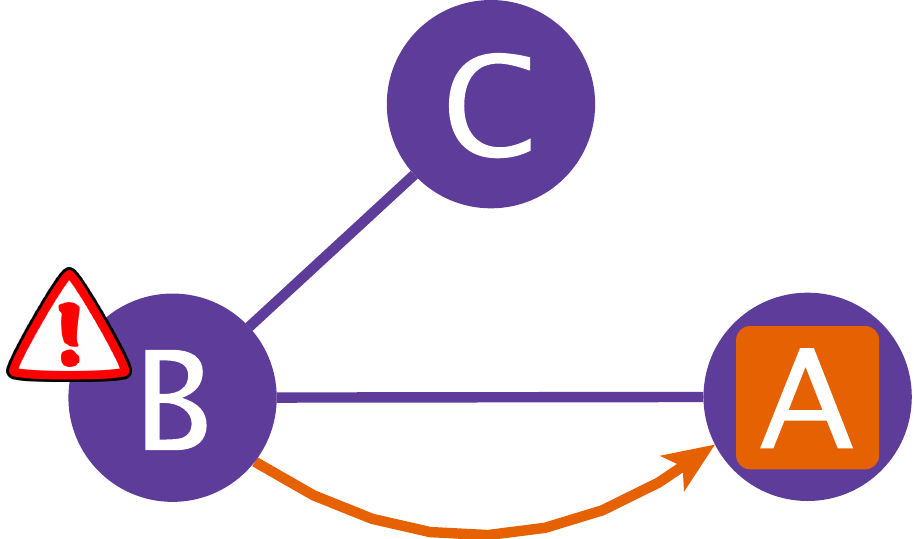}}
	\subfigure[False positive. Undesirable.]{\includegraphics[width=0.24\columnwidth]{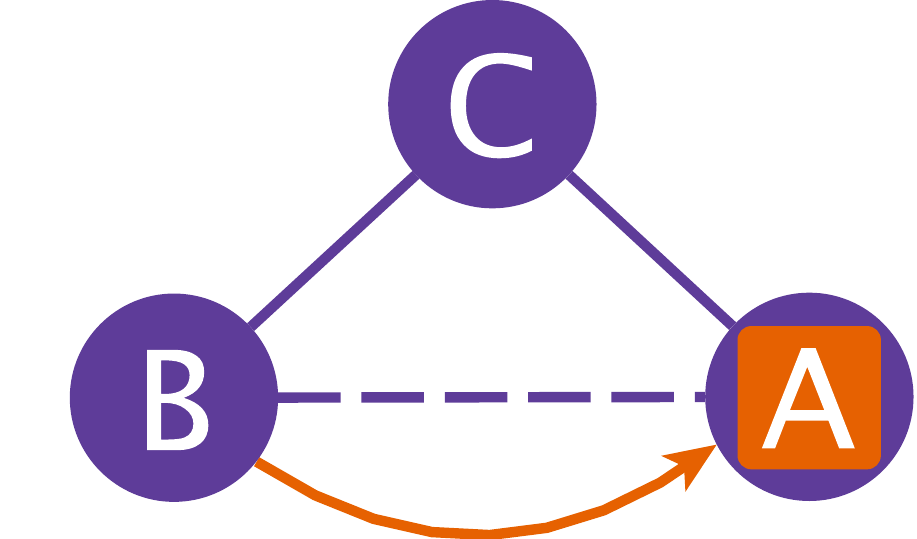}}
	\caption{Self-healing dilemmas in fault detection under uncertainty. In this illustrative example, fault correction is the reactive establishment of the link between Node $A$ and $C$ when no heartbeat messages are received in Node $A$ from Node $B$. In contrast, fault tolerance waits further for the heartbeat messages to arrive assuming a delay over the link between Node $A$ and $B$. Possible outcomes: (a) \emph{True negative}: Parent Node $B$ is healthy and the self-healing agent does not perform fault correction. This outcome has no inconsistency cost. (b) \emph{True positive}: Parent Node $B$ is faulty and the self-healing agent performs fault correction. This outcome has no inconsistency cost. (c) \emph{False negative}: Parent Node $B$ is faulty but the self-healing agent does not perform fault correction. This outcome has inconsistency cost. (d) \emph{False positive}: Parent Node $B$ is healthy but the self-healing agent performs fault correction. This outcome has inconsistency cost. }\label{fig:outcomes}
\end{figure}

The next section formalizes the fault scenarios of false negative and false positive outcomes during system runtime, which are the ones that come with an inconsistency cost. The modeled fault scenarios serve the following: (i) Predict the inconsistency cost of decentralized self-healing systems without application information. (ii) Design self-healing agents with a fault-detection capability that minimizes the inconsistency cost during system runtime. 

\section{Modeling Fault Scenarios at System Runtime}\label{sec:fault-scenarios}

Table~\ref{table:math-symbols} summarizes the mathematical symbols used in the rest of this paper. Assume once more here the pair of Nodes $A$ and $B$, where Node $A$ remotely monitors the health status of Node $B$. Tracking the inconsistency cost generated by this pair of nodes during system runtime is complex and challenging due to the uncertainty over the different \emph{fault scenarios} in the following: (i) Faults in either of the two (monitoring and monitored) nodes. (ii) Faulty detection in Node $A$, performed either too early or too late. Table~\ref{table:scenarios} illustrates the modeled fault scenarios that can occur during system runtime. 

\begin{table}[!htb]
	\caption{Mathematical notation used in this paper.}\label{table:math-symbols}
	\centering
	\resizebox{\columnwidth}{!}{%
		\begin{tabular}{l l}
			\toprule
			Symbol & Description\\
			\midrule
			$T$ & System runtime\\
			$\tau \in \{1, ..., T\}$ & Time unit\\
			$d \leq T$ & Detection time\\
			$t \leq T$ & Threshold for fault detection\\
			$F_{A}, F_{B} \leq T$ & Fault time of Node $A$ and $B$\\
			$l$ & Total number of fault scenarios\\
			$s \in \{1, ..., l\}$ & Fault scenario\\
			$\epsilon_{s}^{-}, \epsilon_{s}^{+} \in \mathbb{R}$ & Maximum inconsistency cost of a fault scenario $s$ generated by\\&a false negative and false positive state, during system runtime $T$\\
			$\rho_{s}^{-}, \rho_{s}^{+} \in [0, 1]$ & Relative inconsistency cost of a fault scenario $s$ generated by\\&a false negative and false positive state, during system runtime $T$\\
			$C, C^{-}, C^{+} \in \mathbb{R}$ & Total inconsistency cost as well as  inconsistency cost generated by\\&a false negative and false positive state during system runtime $T$\\
			$c \leq T$ & Time to restore aggregation accuracy by a corrective operation\\
			$R_{\mathsf{B}} \leq T$ & Recovery time of Node $B$\\
			$p$ & Time to propagate a node descriptor to another node\\
			$n \in \mathbb{N}$ & Total number of nodes\\
			$m \in \mathbb{N}$ & Number of batches of faulty nodes\\
			$k \in \mathbb{N}, m \dot k \leq n$ & Number of faulty nodes at each batch\\
			$\lambda \in [0,1]$ & Calibration factor\\
			$RMSE$ & Root mean square error\\
			$C_{\mathsf{R}}, C_{\mathsf{GR}}$ & Total inconsistency cost of regression and generalized\\&regression calibration methods\\
			$C_{\mathsf{D}}$ & Predicted inconsistency cost of DIAS\\
			\bottomrule
		\end{tabular}
	}
\end{table}

\begin{table}[!htb]
\caption{Modeling fault scenarios between the pair of Nodes $A$ and $B$ during system runtime $T$ and the relative inconsistency cost $\rho_{s}^{}$ by the false positive and false negative outcomes. These scenarios are formalized with Node $A$ monitoring the health status of Node $B$. The scenarios are modeled based on the following information: (i) The health status of the nodes, i.e. which of the two nodes are ON/OFF (healthy/faulty). (ii) The timing of the faults $F_{A}, F_{B}$, i.e. which node is faulty first. Cost calculations rely on the detection time $d$ of Node $A$ and the detection threshold $t$.} \label{table:scenarios}
\centering
\resizebox{\columnwidth}{!}{%
\begin{tabular}{l l l l l}
\toprule
&  \multicolumn{2}{ c }{Fault Scenario ($s$)} & \multicolumn{2}{c}{\makecell[c]{Relative\\Inconsistency Cost ($\rho_{s}^{}$)}}\\\cmidrule{2-3}
\cmidrule{4-5}
Depiction & \makecell[l]{Health\\Status} & \makecell[l]{Faults\\Timing} & \makecell[l]{False\\Positive} & \makecell[l]{False\\Negative}\\\addlinespace\toprule
\begin{minipage}{.25\textwidth}\includegraphics[width=1.0\textwidth]{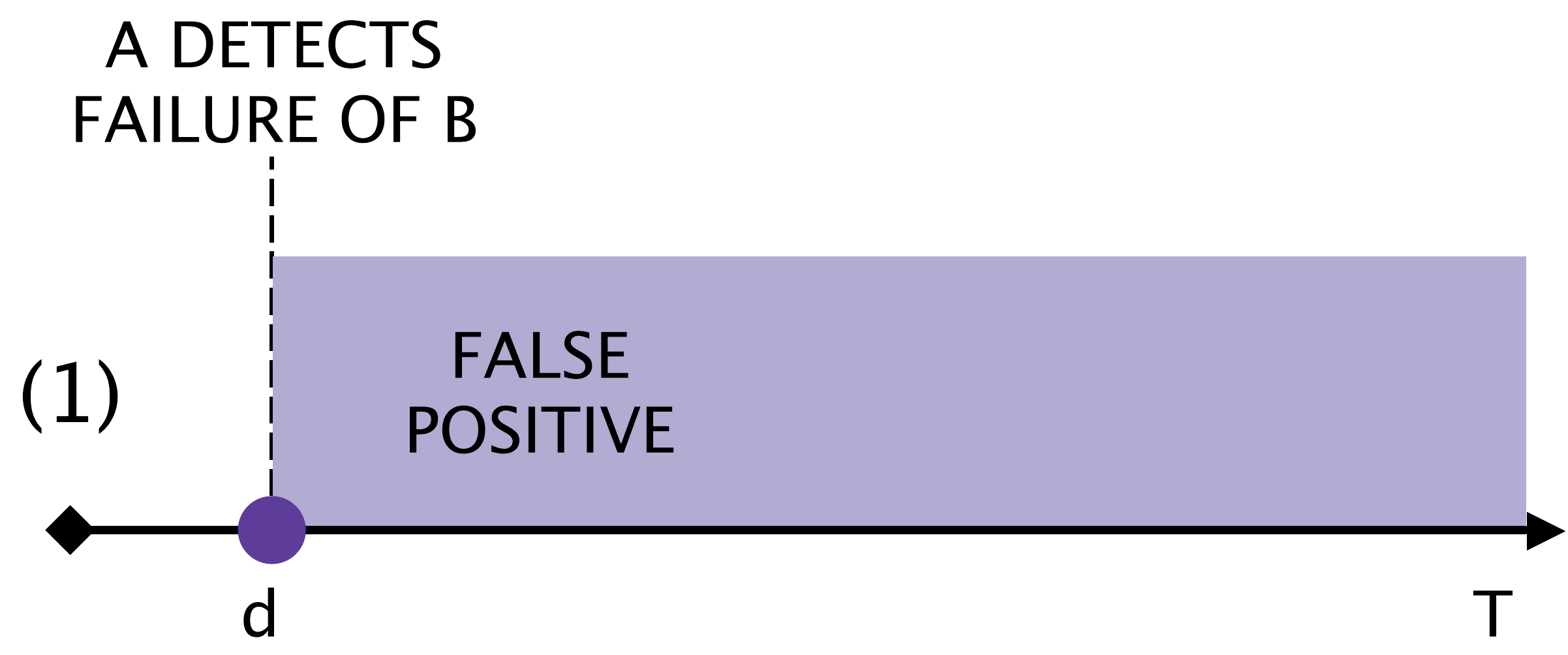}\end{minipage} & \makecell[l]{$A$: ON\\ \\$B$: ON}& - & $\frac{T-d}{T-t}$ & - \\
\begin{minipage}{.25\textwidth}\includegraphics[width=1.0\textwidth]{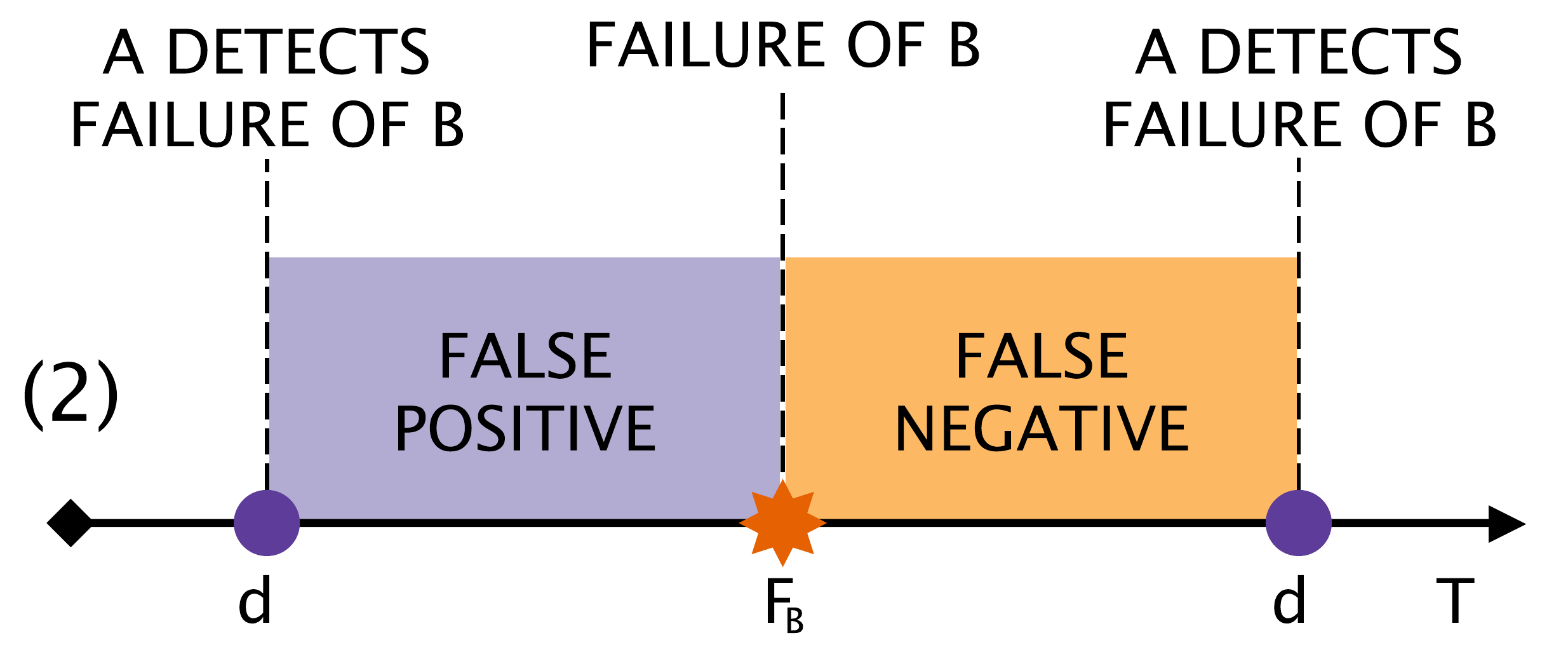}\end{minipage} & \makecell[l]{A: ON\\ \\B: OFF} & $F_{B}$ & $\frac{F_{B}-d}{F_{B}-t}$ & $\frac{d-F_{B}}{T-F_{B}}$ \\
\begin{minipage}{.25\textwidth}\includegraphics[width=1.0\textwidth]{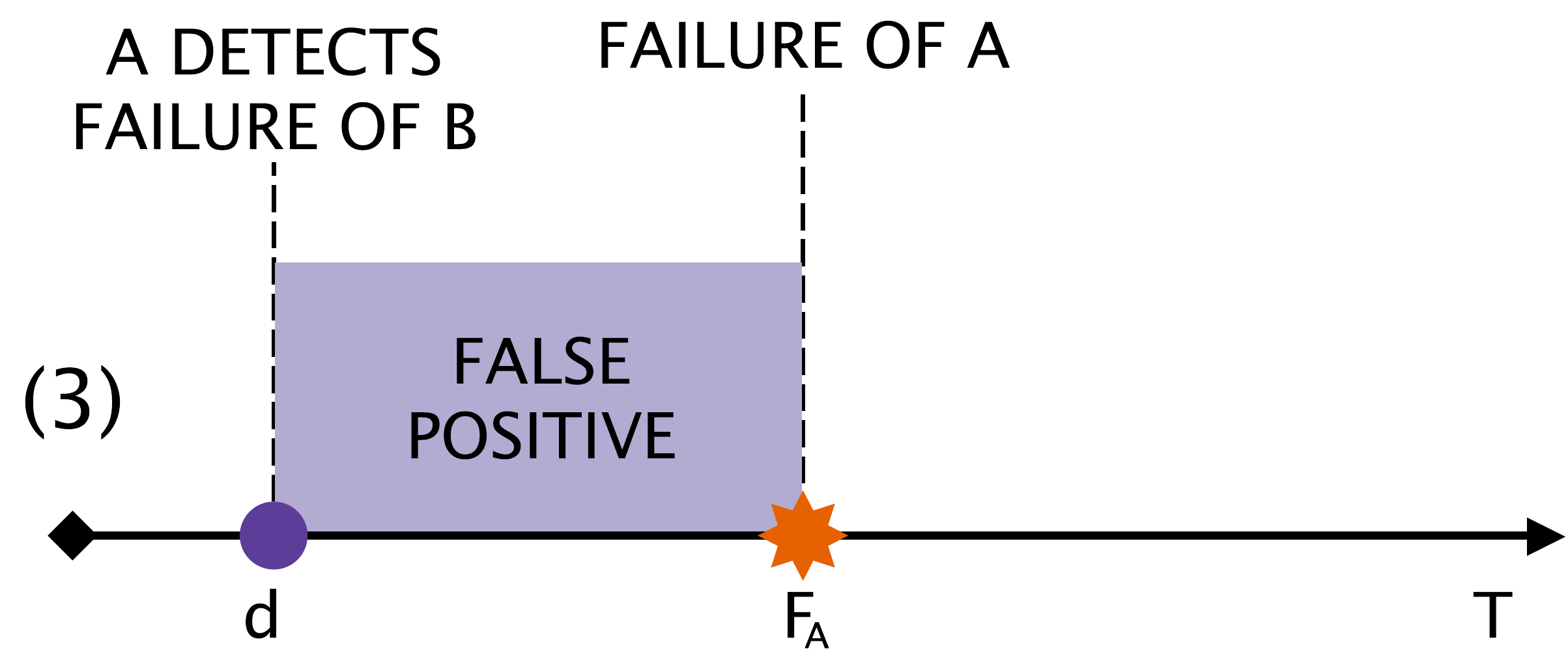}\end{minipage} & \makecell[l]{$A$: OFF\\ \\$B$: ON} & $F_{A}$ & $\frac{F_{A}-d}{F_{A}-t}$ & - \\
\begin{minipage}{.25\textwidth}\includegraphics[width=1.0\textwidth]{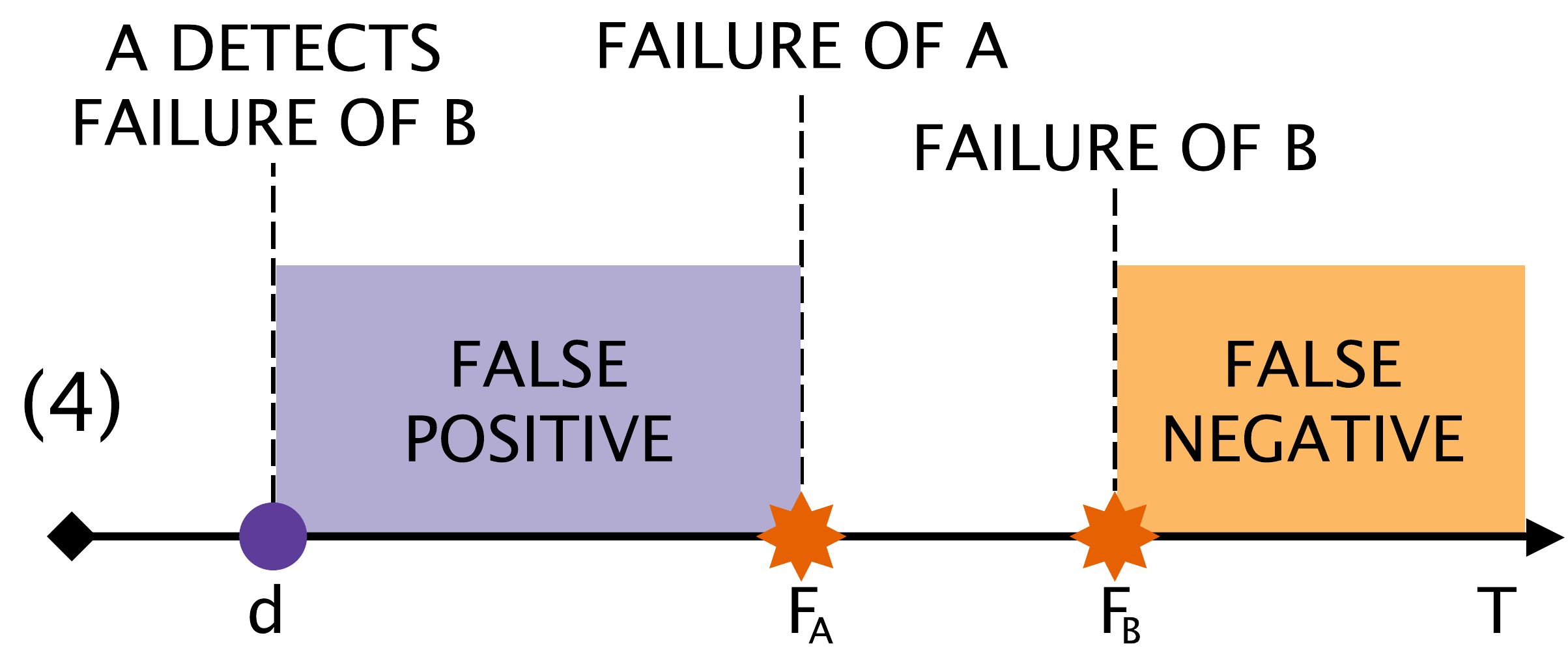}\end{minipage} & \makecell[l]{A: OFF\\ \\B: OFF} &  $F_{A}<F_{B}$ & $\frac{F_{A}-d}{F_{A}-t}$ & $\frac{T-F_{B}}{T-F_{B}}=1$ \\
\begin{minipage}{.25\textwidth}\includegraphics[width=1.0\textwidth]{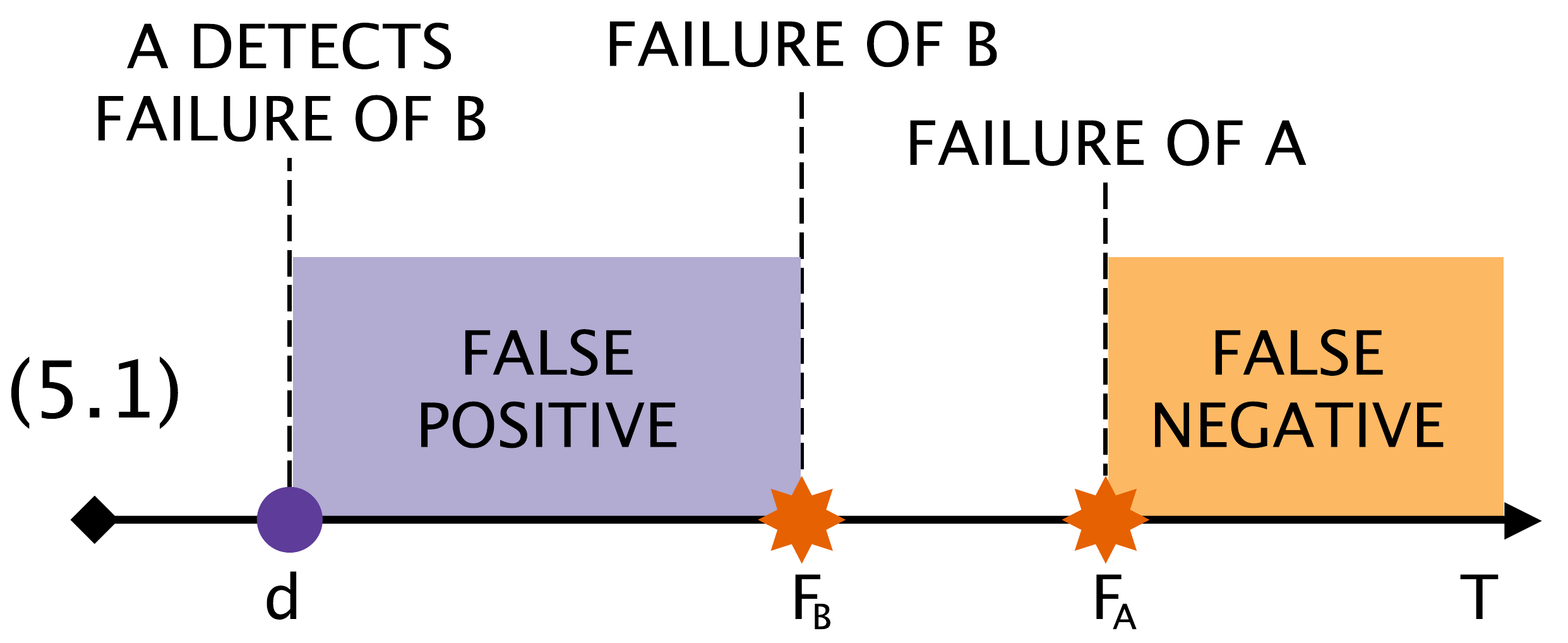}\end{minipage} & \makecell[l]{A: OFF\\ \\B: OFF} & \makecell[l]{$F_{A}>F_{B}$,\\ \\$d<F_{B}$} & $\frac{F_{B}-d}{F_{B}-t}$ & $\frac{T-F_{A}}{T-F_{A}}=1$ \\
\begin{minipage}{.25\textwidth}\includegraphics[width=1.0\textwidth]{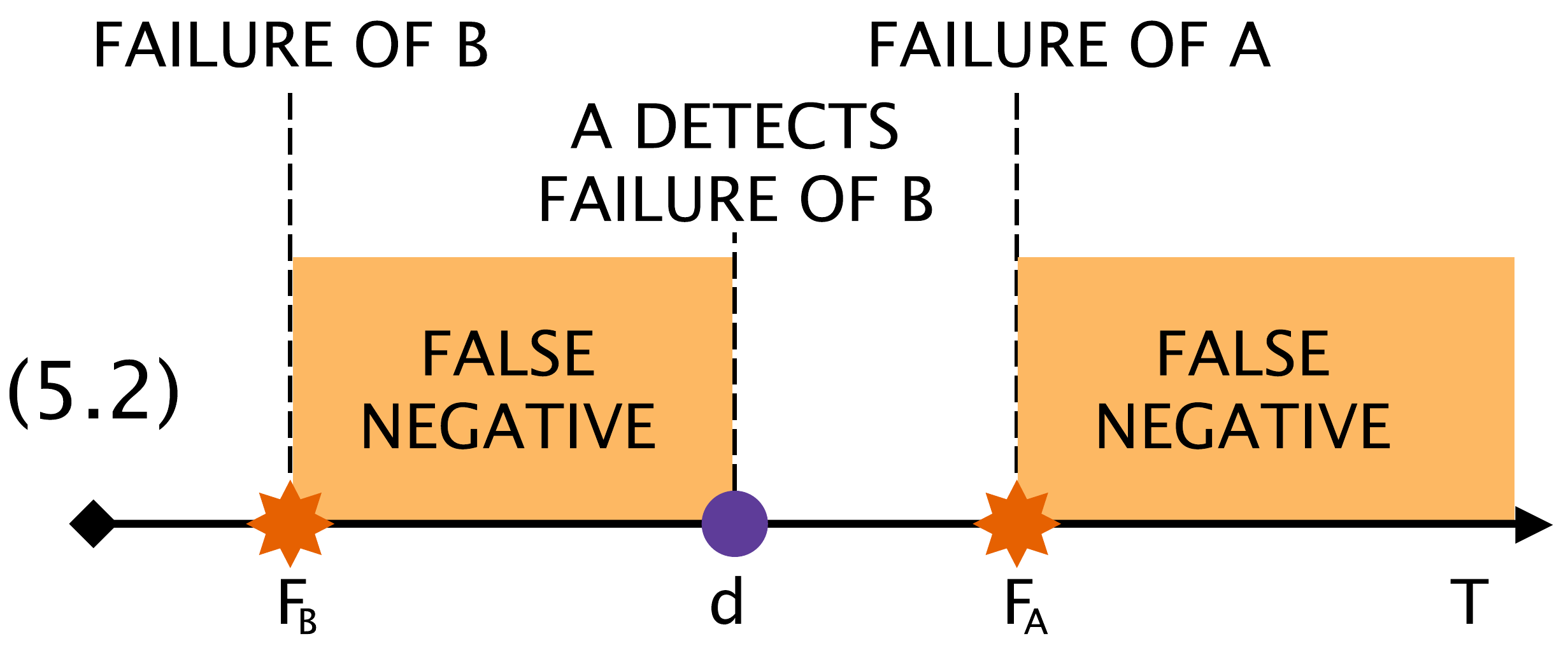}\end{minipage} & \makecell[l]{A: OFF\\ \\B: OFF} & \makecell[l]{$F_{A}>F_{B}$,\\ \\$d>F_{B}$} & - & \makecell[l]{$\frac{d-F_{B}}{F_{A}-F_{B}}$,\\\\$\frac{T-F_{A}}{T-F_{A}}=1$} \\
\begin{minipage}{.25\textwidth}\includegraphics[width=1.0\textwidth]{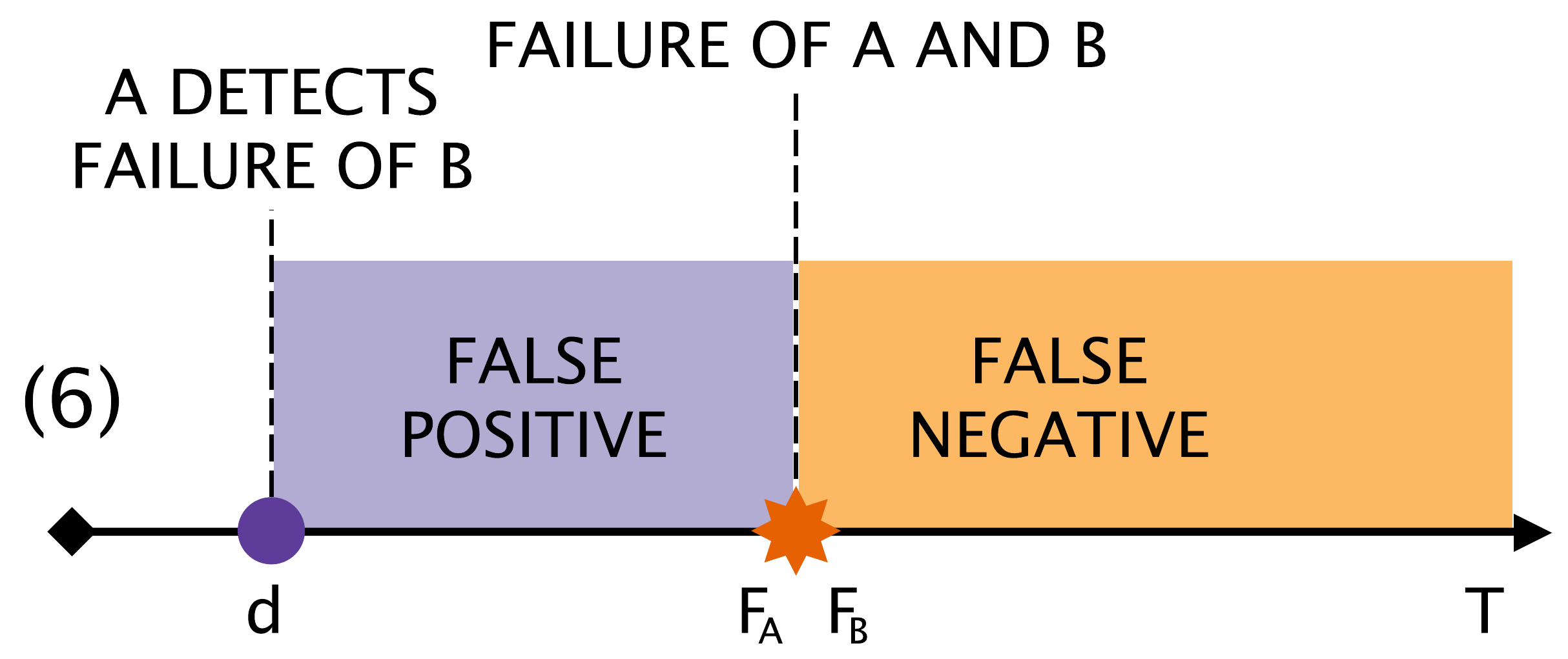}\end{minipage} & \makecell[l]{A: OFF\\ \\B: OFF} & $F_{A}=F_{B}$ & $\frac{F_{B}-d}{F_{B}-t}$ & $\frac{T-F_{A}}{T-F_{A}}=1$ \\\bottomrule
\end{tabular}
}
\end{table}

A \emph{system runtime} $T$ is studied over which the inconsistency cost is traced as well as the \emph{detection time} $d\leq T$ of Node $A$ identifying Node $B$ as faulty. Without loss of generality, Node $A$ triggers at time $d$ fault correction after a \emph{threshold} time period $t \leq T$ during which ($d-t, d-t+1,..., d-1,d$) the fault-detection criterion is satisfied, e.g. a heartbeat message is not received. Otherwise, Node $A$ performs fault tolerance, awaiting further for the heartbeat messages to arrive and relying on underlying preventive maintenance mechanisms for recovery. Both Node $A$ and $B$ can become faulty at any time $F_{A}, F_{B} \leq T$ respectively. At each time $\tau \in \{1,...,T\}$ the nodes can be in one of the states of Figure~\ref{fig:outcomes}. False negative and false positive states generate at each time $\tau$ an inconsistency cost value for a given \emph{fault scenario} $s \in \{1,...,l\}$ out of $l$ possible fault scenarios shown in Table~\ref{table:scenarios}. Moreover, for a fault scenario $s$, the time during which the pair of Nodes $A$ and $B$ are in a false negative or false positive state out of the total time period in which they can be in such a state during system runtime is measured by $\rho_{s}^{-}$ and $\rho_{s}^{+}$ respectively. Therefore, the total inconsistency cost $C$ generated by a pair of agents during the system runtime $T$ can be measured as follows:

\begin{equation}
C=C^{-}+C^{+}=\sum_{s=1}^{l} (\epsilon_{s}^{-}\cdot\rho_{s}^{-} + \epsilon_{s}^{+}\cdot\rho_{s}^{+}),\label{eq:inconsistency-cost}
\end{equation}

\noindent where $\epsilon_{s}^{-}$ and $\epsilon_{s}^{+}$ are the total (maximum) inconsistency cost that can be generated during system runtime by a fault scenario $s$. These can be a result of a constant unit of inconsistency cost value generated at each time point $\tau$ or they can be the output of functions $\epsilon_{s}^{-}=\sum_{\tau=1}^{T}f^{-}(\tau)$, $\epsilon_{s}^{+}=\sum_{\tau=1}^{T}f^{+}(\tau)$ representing an analytical or empirical model~\cite{Zhou2018,Bhandari2018,Gorbenko2019}. The fault scenarios of Table~\ref{table:scenarios} are illustrated as follows:

\begin{itemize}
\item[\textbf{1.}] $A$: ON, $B$: ON -- Both nodes do not defect.\\
The inconsistency cost by a false positive outcome is generated during the time period $T-d$ in which Node $A$ erroneously detects Node $B$ as faulty, while it is not. The maximum time during which nodes can be in this state is $T-t$. 
\item[\textbf{2.}] $A$: ON, $B$: OFF -- Node $B$ becomes faulty at time $F_{B}$, while Node $A$ does not defect. \\
If fault detection occurs before $F_{B}$, the inconsistency cost by a false positive outcome is generated during the time period $F_{B}-d$, with a maximum duration of $F_{B}-t$. If fault detection occurs after $F_{B}$, the inconsistency cost by a false negative outcome is generated during the time period $d-F_{B}$, with a maximum duration of $T-F_{B}$. 
\item[\textbf{3.}] $A$: OFF, $B$: ON -- Node $A$ becomes faulty at time $F_{A}$, while Node $B$ does not defect.\\
If fault detection occurs before $F_{A}$, the inconsistency cost by a false positive outcome is generated during the time period $F_{A}-d$, with a maximum duration of $F_{A}-t$.
\item[\textbf{4.}] $A$: OFF, $B$: OFF -- Node $A$ becomes faulty at time $F_{A}$ and Node $B$ at $F_{B}>F_{A}$.\\
If fault detection occurs before $F_{A}$, the inconsistency cost by a false positive outcome is generated during the time period $F_{A}-d$, with a maximum duration of $F_{A}-t$. The inconsistency cost by a false negative outcome is generated during the time period $T-F_{B}$ that is also the maximum duration.  This is because Node $A$ is faulty to perform fault detection of Node $B$. 
\item[\textbf{5.1}] $B$: OFF, $A$: OFF, $d<F_{B}$ -- Node $A$ becomes faulty at time $F_{A}$ and Node $B$ at $F_{B}<F_{A}$.\\
If fault detection occurs before $F_{B}$, the inconsistency cost by a false positive outcome is generated during the time period $F_{B}-d$, with a maximum duration of $F_{B}-t$. The inconsistency cost by a false negative outcome is generated during the time period $T-F_{A}$ that is also the maximum duration.  This is because Node $A$ becomes faulty and is unable to perform fault correction of Node $B$. 
\item[\textbf{5.2}] $B$: OFF, $A$: OFF, $d>F_{B}$ -- Node $A$ becomes faulty at time $F_{A}$ and Node $B$ at $F_{B}<F_{A}$.\\
The inconsistency cost by a false negative outcome is generated during the time period $d-F_{B}$, with a maximum duration of $F_{A}-F_{B}$. This represents the lag of detecting the faulty Node $B$. There is also an additional inconsistency cost by a false negative outcome during the time period $T-F_{A}$ that is also the maximum duration. This is because Node $A$ becomes faulty and is unable to perform fault correction of Node $B$. 
\item[\textbf{6.}] $A$, $B$: OFF -- Node $A$ and $B$ become faulty at time $F_{A}=F_{B}$.\\
The inconsistency cost by a false positive outcome is generated during the time period $F_{B}-d$, with a maximum duration of $F_{B}-t$. The inconsistency cost by a false negative outcome is generated during the time period $T-F_{A}$ that is also the maximum duration. This is because Node $A$ becomes faulty and is unable to perform fault correction of Node $B$. 
\end{itemize}

As proven below, these fault scenarios are \emph{sufficient} to model the overall system health status:

\begin{theorem}\label{th:health-status-space} 
The fault times $F_{A}$, $F_{B}$ of each possible pair of a Node $A$ monitoring the health status of Node $B$ are sufficient to calculate the health status of a decentralized system of $n$ nodes that arbitrary defect in $m$ batches, each of size $k<n$. 
\end{theorem}
\begin{proof}
The size of the health status space in a decentralized system of $n$ nodes, where each node monitors the health status of all other nodes, is $n^{2}-n$ as each node monitors the health status of all other $n-1$ nodes. The fault times $F_{A}$, $F_{B}$ $\neq 0$ and the six fault scenarios they determine (outlined in Table~\ref{table:scenarios}) are sufficient to calculate this space, if the number of node pairs determined at each fault scenario sum up to $n^{2}-n$. For each fault scenario, the number of node pairs are determined as follows:
\begin{itemize}
	\item $A$: ON, $B$: ON, $F_{A}=F_{B}=0$ \\
	Each of the $n-mk$ healthy nodes monitors the health status of all other $n-mk-1$ healthy nodes, i.e. nodes do not monitor their own health status. This is the number of healthy node pair variations without repetition:
		\begin{equation}\label{eq:fault-scenario-1}
			\binom {n-mk}{2} 2! = \frac{(n-mk)!}{2!(n-mk-2)!} = (n-mk)(n-mk-1)
		\end{equation}
		
	\item $A$: ON, $B$: OFF, $F_{A}=0<F_{B}\leq T$ and \\
	$A$: OFF, $B$: ON, $F_{B}=0<F_{A}\leq T$ \\
	Each of the $n-mk$ healthy nodes monitors the health status of all other $mk$ faulty nodes and vice versa for these two fault scenarios. The total number of these node pairs for both fault scenarios are calculated as: 
		\begin{equation}\label{eq:fault-scenario-2-3}
			2mk(n-mk)
		\end{equation}
		
	\item $A$: OFF, $B$: OFF, $0<F_{A}<F_{B}\leq T$ and\\
	$B$: OFF, $A$: OFF, $0<F_{B}<F_{A}\leq T$ \\
	These fault scenarios involve $k^{2}$ pairs of faulty Nodes $A$ and $B$ for each possible pair of different batches of faulty nodes. The total number of all possible batch pairs in which a pair of two faulty Nodes $A$ and $B$ reside is the number of batch pair variations without repetition. Therefore, the total number of these faulty node pairs for both fault scenarios is calculated as follows: 
		\begin{equation}\label{eq:fault-scenario-4-5}
			k^{2} \binom{m}{2}2! = k^{2} \frac{m!}{2!(m-2)!}2! = mk^{2}(m-1)
		\end{equation}
		
	\item $A$, $B$: OFF, $0<F_{A}=F_{B} \leq T$ \\
	This fault scenario involves $k^{2}-k$ pairs of faulty Nodes $A$ and $B$ defecting at the same time at each of the $m$ batches of faulty nodes. This is the number of faulty node pair variations without repetition: 
		\begin{equation}\label{eq:fault-scenario-6}
			m\binom{k}{2}2! = m\frac{k!}{2!(k-2)!}2!=mk(k-1)
		\end{equation}
\end{itemize}
The number of node pairs from all fault scenarios sum up as follows:
\begin{align}
	\overbrace{(n-mk)(n-mk-1)}^{\text{Fault scenario 1, Eq.~\ref{eq:fault-scenario-1}}}	+&  \nonumber\\
	\overbrace{2mk(n-mk)}^{\text{Fault scenario 2 \& 3, Eq.~\ref{eq:fault-scenario-2-3}}} +  &\nonumber\\
	\overbrace{mk^{2}(m-1)}^{\text{Fault scenario 4 \& 5, Eq.~\ref{eq:fault-scenario-4-5}}} + &  \nonumber\\
	\overbrace{mk(k-1)}^{\text{Fault scenario 6, Eq.~\ref{eq:fault-scenario-6}}} =& \text{ } n^{2}-n, 
\end{align}
which is the overall health status space of a decentralized system. 
\end{proof}
\section{Model Applicability}\label{sec:model-applicability}

This section illustrates the applicability of the fault scenarios for self-healing in decentralized in-network data aggregation. Fault detection, fault correction and fault tolerance are illustrated. 

\subsection{Fault detection via gossip-based communication}\label{subsec:fault-detection}

Gossip-based communication~\cite{Shah2009} is selected for fault detection given the following: (i) It is a communication protocol for large-scale and highly decentralized systems that falls within the scope of this paper. (ii) It is general-purpose and fundamental as it can be widely used for fast information dissemination, new information discovery, preserving network robustness by keeping the network connected, and other core operations required in decentralized systems~\cite{Jelasity2007,Shah2009,Snyder2012}. (iii) It finds real-world applicability in several systems such as peer-to-peer networks~\cite{Jelasity2007}, cloud computing~\cite{Marzolla2011,Lim2018}, Big Data systems~\cite{Cao2018}, distributed ledgers~\cite{He2019,Baird2016}, middleware systems~\cite{Preisler2015} etc. (iv) It is probabilistic in nature and as a result, fault detection based on gossiping communication comes with uncertainties within which dilemmas of fault correction vs. fault tolerance can be systematically studied. 

Gossip-based communication realizes health status monitoring as illustrated in Figure~\ref{fig:concept}c. Nodes execute a gossiping protocol such as the peer sampling service~\cite{Jelasity2007} that equips each node with a limited-size list of node descriptors, each containing the IP address, port number, timestamp and application information. This list is the partial view that nodes have of the system. It is periodically updated with new random node descriptors during peer-to-peer gossip exchanges with other random nodes selected from the partial view (the same list). 

The health status of the parent node is locally determined by the time period passed since the last time the descriptor\footnote{Descriptors of the parent node with a timestamp value later than the migration time are the ones counted. Earlier descriptors of the parent node may be present and circulated in the network. They are eventually replaced with the latest one during the gossip exchanges~\cite{Jelasity2007}.} of the parent node was present in the partial view of the node in which the self-healing agent resides. If the threshold $t$ is surpassed, the parent node is considered faulty and fault correction is initiated. Otherwise, fault tolerance is performed by awaiting further for the parent node descriptor to arrive by relying on underlying gossip-based communication.

An effective choice of the threshold $t$ depends on the system size and the internal configuration of the gossip-based communication protocol: (i) The size of the partial view. (ii) The execution period. (iii) The node and view selection policy that determine the level of randomization in the communication and exchange of node descriptors respectively. The threshold choice also depends on the external environment, e.g. latency, convergence speed of the communication model, bandwidth and load of the network~\cite{Lavinia2011}. Even if all these uncertainties that determine whether the parent node is truly faulty or not are controlled, the dilemma of the self-healing agent remains: is it fault tolerance or fault correction that results in lower inconsistency cost? Given that the inconsistency cost is context/application dependent, this paper introduces the computational case study of decentralized data aggregation within which inconsistency cost is assessed.

\subsection{Computational case study: decentralized data aggregation}\label{subsec:DIAS}

The computational problem of dynamic in-network data aggregation is studied~\cite{Fasolo2007}. More specifically, this paper studies how self-healing can improve the accuracy in decentralized computations of aggregation functions when nodes fail. The computational case study is the following: Each node in the network is a \emph{data supplier} and \emph{data consumer} (extreme performance benchmark). Data suppliers generate and share data (streams) with data consumers. Data consumers collect data (streams) from data suppliers and compute/update aggregation functions such as average, summation, count, maximum/minimum and other. When a data supplier disconnects from the network, data consumers need to update their aggregation function by performing a reverse computation, i.e. rollback, that removes the counted input data of the departing data supplier. 

Preserving accurate estimations of aggregation functions in this computational case study is challenging given that (i) data suppliers and consumers need to discover each other in a decentralized unstructured network, (ii) data suppliers can change the input data of the aggregation functions, (iii) data consumers may compute any aggregation function given the input data of data suppliers and (iv) reverse computations are required when data suppliers leave the network. In contrast to earlier decentralized aggregation methodologies such as gossiping~\cite{Jelasity2005,Pianini2016}, tree-based~\cite{Ding2003} or synopsis diffusion~\cite{Nath2008}, DIAS\footnote{Available at \url{http://dias-net.org} (last accessed: March 2021).}, the \emph{Dynamic Intelligent Aggregation Service}~\cite{Pournaras2017,Pournaras2017b,Pournaras2017c} is a decentralized gossip-based aggregation system designed to meet all these requirements\footnote{This is made possible by using an efficient and scalable distributed memory system based on probabilistic data structures, the Bloom filters~\cite{Bloom1970}. Based on Bloom filters, a data supplier can reason whether it has earlier communicated with a data consumer to share data and vice versa a data consumer can reason whether is has earlier communicated with a data supplier to aggregate data. Data suppliers and consumers can also reason about what data have been shared and aggregated, i.e. the most recent ones or outdated ones, so that aggregation inaccuracies are minimized, while unnecessary communication is limited. Further information about DIAS is out of the scope of this paper and can be found in earlier work~\cite{Pournaras2017,Pournaras2017b,Pournaras2017c} .} and therefore it is used to assess how well the inconsistency cost of the fault scenarios predicts the aggregation inaccuracies.

The inconsistency cost is measured by the average relative approximation error in the estimation of the aggregation functions among all data consumers in the network. In other words, the inconsistency cost measures how far the estimation of the aggregates is from the actual true values of the aggregates. Apparently, when nodes hosting data suppliers become faulty, reversed (rollback) computations are required by data consumers that have earlier aggregated data of these now faulty data suppliers. Without such computations, the estimations of the aggregates diverge from the actual ones generating inconsistency cost (false negative state in Figure~\ref{fig:outcomes}c). However, inconsistency cost may also result by reversed (rollback) computations because of an erroneous gossip-based fault detection, e.g. a very low threshold value $t$ that determines the node hosting the data supplier as faulty when actually it is not (false positive state in Figure~\ref{fig:outcomes}d). Therefore, the self-healing dilemma is highly applicable in this computational case study and the rest of this section introduces the functionality of the fault correction and fault tolerance in DIAS.

\subsection{Fault-corrective aggregation}\label{subsec:fault-correction}

This paper extends an earlier self-corrective aggregation mechanism~\cite{Pournaras2017c} for nodes joining and leaving the network into a fault-correction mechanism when nodes arbitrary fail. The rationale of self-correction when a node with a data supplier leaves the network is the following: A self-healing agent creates a replica of the data supplier with which it migrates to a remote random neighboring host node (see Figure~\ref{fig:concept}b) selected via the peer sampling service\footnote{Random selection of the migration host is performed for load-balancing. Without loss of generality, DIAS reuses the peer sampling service for the purpose of the migrations to limit the need for another such mechanism that comes with additional performance overhead, i.e. communication, processing and storage cost. Other methodologies for migration include random walks in the network or allocating dedicated nodes for redundancy~\cite{Oyediran2016}.} based on which DIAS operates. The migrated data supplier initiates corrective rollback operations with the data consumers in the network to update the aggregation functions. This process either completes or is interrupted when the self-healing agent detects\footnote{The returned parent node is detected when its descriptor appears in the partial view of the migrated node with a timestamp value later than the leave.} via the peer sampling service that the parent node has joined again the network. In the latter case, the migrated self-healing agent together with the migrated data supplier return back to the parent node to continue their operations as before. Migrations can be consecutive if the migrated host node leaves the network as well. More information about the protocol specification and evaluation results can be found in earlier work~\cite{Pournaras2017c}. 

The limitation of this mechanism is that self-corrective operations are initiated reactively by the parent node before leaving the network. This is not realistic in a scenario of arbitrary node failures that can terminate all local processes before self-corrective operations are initiated. This paper extends this model by proactively migrating each self-healing agent to a remote host, where it runs as a daemon monitoring the health status of the parent node as shown in Figure~\ref{fig:concept}c. Monitoring is performed by reusing the peer sampling service\footnote{Other mechanisms such as heartbeat messages~\cite{Hasan2018,Gyamfi2019} can be used.} according to the fault-detection mechanism introduced in Section~\ref{subsec:fault-detection} so that no other performance overhead is introduced.

\subsection{Fault-tolerant aggregation}\label{subsec:fault-tolerance}

The alternative to fault correction is fault tolerance that determines no corrective operations until the threshold $t$ is reached. Fault tolerance eliminates inconsistency costs originated by false positive states (see Figure~\ref{fig:outcomes}d). Moreover, fault tolerance is cost-effective when the faulty node can recover promptly, given the time required for corrective operations to complete. More specifically, fault tolerance eliminates inconsistency cost if it holds: 
\vspace{-0.1cm}
\begin{equation}
F_{B}+t+c>R_{B}+p
\end{equation}\label{eq:fault-tolerance}

\noindent where $F_{B}$ is the time when Node $B$ becomes faulty, $t$ is the fault-detection threshold and $c$ is the duration for the corrective operations to restore a required aggregation accuracy level. On the other side of the inequality, $R_{B}$ is the time when Node $B$ recovers\footnote{The scenario in which $F_{B}+t+c<R_{B}+p$ is more complex to determine whether fault tolerance or fault correction should be performed as it depends on the relation of $t$ and $p$, the data consumers with which corrective operations have been performed and their aggregated data.} and $p$ is the time required by Node $A$ to detect the recovery, i.e. propagation time of the Node $B$ descriptor by the peer sampling service. This inequality can be used to determine threshold values $t$ for each node given empirical models for $R_{B}-F_{B}$, which are though not the focus of this paper. Instead, different threshold values and their influence on inconsistency cost are studied. 

\section{Experimental Methodology}\label{sec:methodology}

This study has the following three objectives: (i) Profiling of the inconsistency cost generated by the modeled fault scenarios under varying fault scales, fault profiles and fault-detection thresholds. (ii) Validation of whether the inconsistency cost of the modeled fault scenarios is a good general predictor of the accuracy observed in the application scenario of decentralized aggregation of real-world power consumption data. (iii) Comparison of different model calibrators for the prediction of aggregation accuracy. Table~\ref{settings} outlines the experimental parameterization\footnote{The system parameterization of the peer sampling service and DIAS is chosen based on earlier experimental findings~\cite{Jelasity2007,Pournaras2017,Pournaras2017c,Pournaras2017b} and on the rationale of a cost-effective operation of the decentralized data aggregation under no faulty nodes. In this way, the effect of the faults on the data aggregation and how this effect can be controlled via a self-healing tuning (choice of threshold) can be isolated and studied systematically.}. All studied systems are implemented with an improved version~\cite{Fanitabasi2020} of the Protopeer prototyping toolkit~\cite{Galuba2009} for distributed systems.

\begin{table}[!htb]
\caption{An overview of the experimental parameterization.}\label{settings}
\centering
\resizebox{\columnwidth}{!}{%
\begin{tabular}{l l l l}
\toprule
System Parameter & Value & System Parameter & Value \\\midrule
ECBT data\footnotemark[10]~\cite{Pournaras2017} & Day 199 (January $\nth{4}$)  & DIAS execution period & $1 s$ \\
Num. of nodes & 3000 & Num. of aggregation sessions~\cite{Pournaras2017} & 4 \\
Num. of epochs & 3200 & Partial view size~\cite{Jelasity2007} & 50\\
Epoch duration & $250 ms$ & Swap parameter~\cite{Jelasity2007} & 24 \\
Fault scales & $10\%,20\%,...,80\%$ & Healer parameter~\cite{Jelasity2007} & 1\\
Fault profiles (Table~\ref{tab:profiles}) & \nth{1},  \nth{2},  \nth{3} & Fault detection threshold & $[100,800]$ with step 25\\
Num of epochs for bootstrapping &  400 &&\\
\bottomrule
\end{tabular}
}
\end{table}

The following scales of faulty nodes are studied: $\{10\%,20\%,...,80\%\}$ of the total number. Three fault profiles are introduced that come with 1, 2 and 4 batches of faulty nodes respectively: (i) \emph{\nth{1} profile}: All faulty nodes defect in one batch on half of system runtime that is on the $\nth{1600}$ epoch. (ii) \emph{\nth{2} profile}: Faulty nodes defect in two batches, with half of the faulty nodes defecting on the $\nth{1332}$ epoch and the other half on the $\nth{2264}$ epoch. (iii) \emph{\nth{3} profile}: Faulty nodes defect in four batches of equal size on the $\nth{1060}$, $\nth{1620}$, $\nth{2180}$ and $\nth{2740}$ epoch. Such parameters can accurately model failures observed in real-world systems, i.e. failure bursts correlated in time/space~\cite{Gallet2010,Kondo2010,Ghahremani2020}, while the evaluated parameter space with extreme fault scales stretches the experimental evaluations. Table~\ref{tab:profiles} summarizes the applicability of the three fault profiles to the modeled fault scenarios. Note in particular that nodes can defect in any of these batches except the case of the \nth{4} and \nth{5} fault scenarios that determine faulty nodes in different batches, while either Node $A$ or $B$ defects first. In the  \nth{2} profile, Node $A$ cannot defect at the \nth{2} batch if it defects first and respectively, Node $B$ cannot defect at the \nth{1} batch if it defects second. Similarly in the \nth{3} profile, the node that defects first cannot defect at the \nth{4} batch and the node that defects second cannot defect at the \nth{1} batch.

\begin{table}[!htb]
	\caption{Applicability of three fault profiles to each false positive (FP) and false negative (FN) state of the fault scenarios. The frequencies of fault-scenarios sum up to $n^{2}-n$ (Theorem~\ref{th:health-status-space}), where $n$ is number of nodes in the network, $k$ the number of faulty nodes at each batch of defected nodes out of a total of $m$ batches.}\label{tab:profiles}
	\centering
	\resizebox{\columnwidth}{!}{%
		\begin{tabular}{l l l l l l l l l}
			\toprule
			& \multirow{2}{*}[-0.3em]{Health} & \multirow{2}{*}[-0.3em]{Frequency} & \multirow{2}{*}[-0.3em]{State} & \multirow{2}{*}[-0.3em]{Node} & \multicolumn{3}{c}{Defect Batch IDs}\\\cmidrule{6-8}
			& & & & & \makecell[c]{\nth{1} Profile\\$m=1$} & \makecell[c]{\nth{2} Profile\\$m=2$} & \makecell[c]{\nth{3} Profile\\$m=4$} \\\midrule
			1. & \makecell[l]{$A$: ON\\ $B$: ON} & $(n-mk)(n-mk-1)$ & FP & None & \xmark & \xmark & \xmark \\\cmidrule{1-8}
			2. & \makecell[l]{$A$: ON\\$B$: OFF} & $mk(n-mk)$ & FP, FN & $B$ & \makecell[l]{1} & \makecell[l]{1, 2} & \makecell[l]{1, 2, 3, 4} \\\cmidrule{1-8}	
			3. & \makecell[l]{$A$: OFF\\$B$: ON} & $mk(n-mk)$ & FP & $A$ & \makecell[l]{1} & \makecell[l]{1, 2} & \makecell[l]{1, 2, 3, 4}\\\cmidrule{1-8}
			\multirow{2}{*}[-0.3em]{\makecell[l]{4.}} & \multirow{2}{*}[-0.3em]{\makecell[l]{$A$: OFF\\$B$: OFF}} & \multirow{2}{*}[-0.3em]{$\frac{1}{2}mk^{2}(m-1)$} & \multirow{2}{*}[-0.3em]{FP, FN} & $A$ & \xmark & \makecell[l]{1} & \makecell[l]{1, 2, 3}\\\cmidrule{5-8}& & & & $B$ & \xmark & \makecell[l]{2} & \makecell[l]{2, 3, 4} \\\cmidrule{1-8}
			\multirow{2}{*}[-0.3em]{\makecell[l]{5.}} & \multirow{2}{*}[-0.3em]{\makecell[l]{$B$: OFF\\$A$: OFF}} & \multirow{2}{*}[-0.3em]{$\frac{1}{2}mk^{2}(m-1)$} & \multirow{2}{*}[-0.3em]{FP, FN} & $B$ & \xmark & \makecell[l]{1} & \makecell[l]{1, 2, 3} \\\cmidrule{5-8} & & & & $A$ & \xmark & \makecell[l]{2} & \makecell[l]{2, 3, 4} \\\cmidrule{1-8}
			6. & $A$, $B$: OFF & $mk(k-1)$ & FP, FN & $A$, $B$ & \makecell[l]{1} & \makecell[l]{1, 2} & \makecell[l]{1, 2, 3, 4} \\
			\bottomrule		
		\end{tabular}
	}
\end{table}

To address the first objective, the fault profiles are applied to a decentralized network of 3000 nodes each running fault detection with the peer sampling service~\cite{Jelasity2007} as illustrated in Section~\ref{subsec:fault-detection} and with the respective parameters of Table~\ref{settings}. The threshold values of $\{100,125,150,...,775,800\}$ epochs are evaluated. By knowing which nodes defect at which time point during system runtime, all false positive and false negative states in the six possible fault scenarios of Table~\ref{table:scenarios} can be measured and analyzed. This analysis is performed exhaustively to profile the inconsistency cost across three dimensions: $8 \text{ \emph{fault scales}} \times 3 \text{ \emph{fault profiles}}\times 29 \text{ \emph{thresholds}}=696 \text{ \emph{experimental settings}}$.  

The finest-grain measurements of inconsistency cost are performed with size $3000^{2}-3000=8997000$ according to Theorem~\ref{th:health-status-space}: every node monitors the health status of every other node in the network. Given the fault scale ($k$) and fault profile ($m$), the health status of all node pairs is calculated according to the equations of Table~\ref{tab:profiles} (Theorem~\ref{th:health-status-space}) and these calculations result in the relative frequencies of Figure~\ref{fig:pairs-per-scenario-scale}. 

\begin{figure}[!htb]
\centering
\includegraphics[width=\columnwidth]{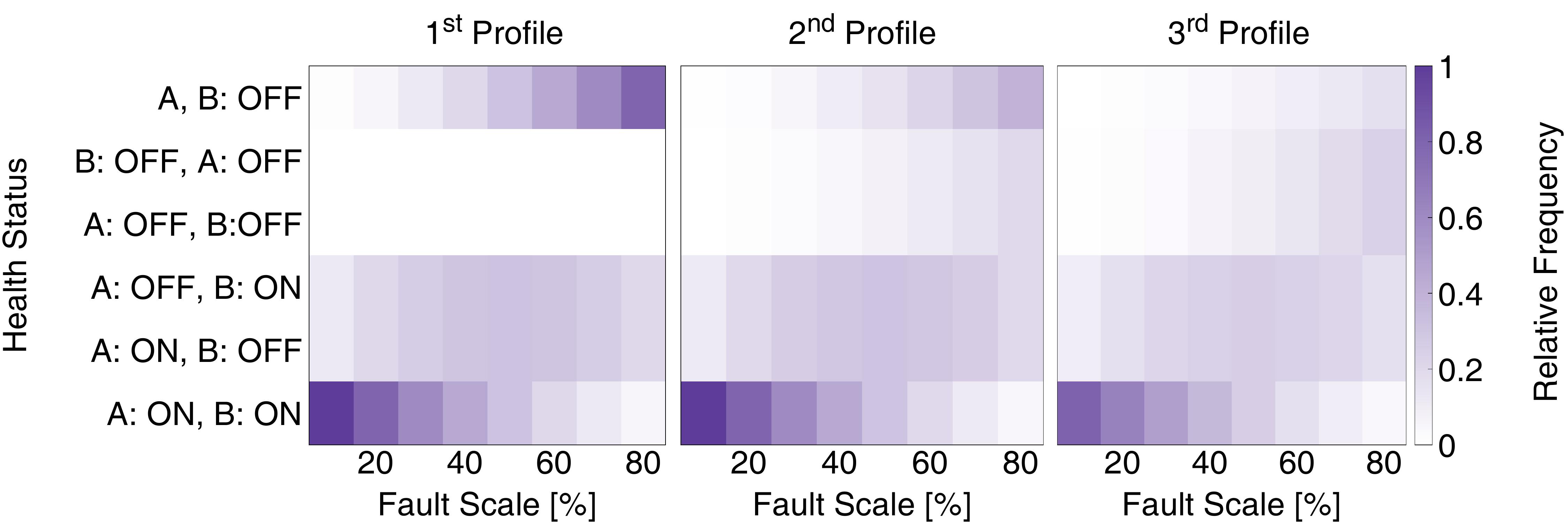}
\caption{Relative frequencies of the health status among all node pairs for different fault scales and fault profiles.}\label{fig:pairs-per-scenario-scale}
\end{figure}

To address the second objective, the modeled fault scenarios are evaluated by measuring how well they predict the inconsistency cost of a self-healing computational/application scenario with faulty nodes. This scenario is the decentralized aggregation of DIAS in which self-healing is performed in terms of fault correction (executing self-corrective operations, see Section~\ref{subsec:fault-correction}) and fault tolerance (postponing self-corrective operations, see Section~\ref{subsec:fault-tolerance}). The prediction of the inconsistency cost is the prediction of the aggregation accuracy measured by the \emph{average relative approximation error} between the \emph{estimated} aggregate values and the \emph{actual} aggregates. In other words, this paper assesses for all 696 experimental settings how good predictor the total inconsistency cost (Equation \ref{eq:inconsistency-cost}) is of the average relative approximation error of DIAS measured over all nodes and throughout system runtime. 

The experiments focus on the summation (total power load) of real-world power consumption data\footnote{Avalable at http://www.ucd.ie/issda/data/commissionforenergyregulationcer/ (last access: March 2021)} from ECBT, the \emph{Electricity Customer Behavior Trial} during 2009-2010 in Ireland. They are collected from smart meters with a frequency of 30 minutes. The power records of the \nth{199} day (4.1.2009) are used for the experiments that are $2 \text{ records/hour} \times 24 \text{ hours}=48 \text{ records}$ uniformly distributed over the system runtime of 2800 epochs, plus 400 epochs for system bootstrapping. Out of the total of 6435 residential and small-medium enterprise consumers in the dataset, 3000 residential consumers are mapped to the 3000 nodes of the decentralized network. Each operates as both data supplier and consumer to evaluate the most demanding computational scenario in which every node shares and aggregates power consumption data. 

Predicting the DIAS accuracy is highly challenging given that the modeled fault scenarios are totally agnostic of the applied (i) computational problem, i.e. aggregation, (ii) algorithm, i.e. DIAS and (iii) data, i.e. power consumption. As such, it is assumed that all fault scenarios $l$ can generate during runtime a total (maximum) inconsistency cost of $\epsilon_{s}^{-}=\epsilon_{s}^{+}=1$. To improve prediction, three model calibration methods are applied that rely on the profiling of inconsistency cost calculated for the first objective of this study. All three calibration methods use application-independent features. One of these methods is totally agnostic of any information about the aggregation problem or the DIAS algorithm, while the other two use DIAS performance target data for fitting a model. Therefore, significant comparisons can be made between a non-calibrated prediction vs. calibrated predictions as well as the application-agnostic calibrations vs. the ones that fit a model to the data. 


\subsubsection{False negative calibration}\label{subsec:calibration-corrective-weighted-average}

The fault scenarios of Table~\ref{table:scenarios} with a false negative state given by $\frac{T-F_{A}}{T-F_{A}}=1$ assume that the fault of Node $A$ generates inconsistency cost throughout the time period $T-F_{A}$ as the fault of Node $B$ cannot be anymore detected (and corrected) during this period. However, recovery may occur earlier, which means in practice that the inconsistency cost may be eliminated within a short period of time, for instance, self-corrective operations in DIAS converge in a finite time period~\cite{Pournaras2017c}. Therefore, this calibration method introduces the calibration factor $\lambda \in [0,1]$ as an additional coefficient for these fault scenarios with false negative state. For $\lambda=1$, no calibration is performed. For each fault scale, the $\lambda$ value with the lowest root mean square error between the predicted and the DIAS inconsistency cost is selected for the comparison with the other calibration methods as shown in Figure~\ref{fig:parameter-configuration}a. 

\begin{figure}[!htb]
\centering
\subfigure[]{\includegraphics[width=1.0\columnwidth]{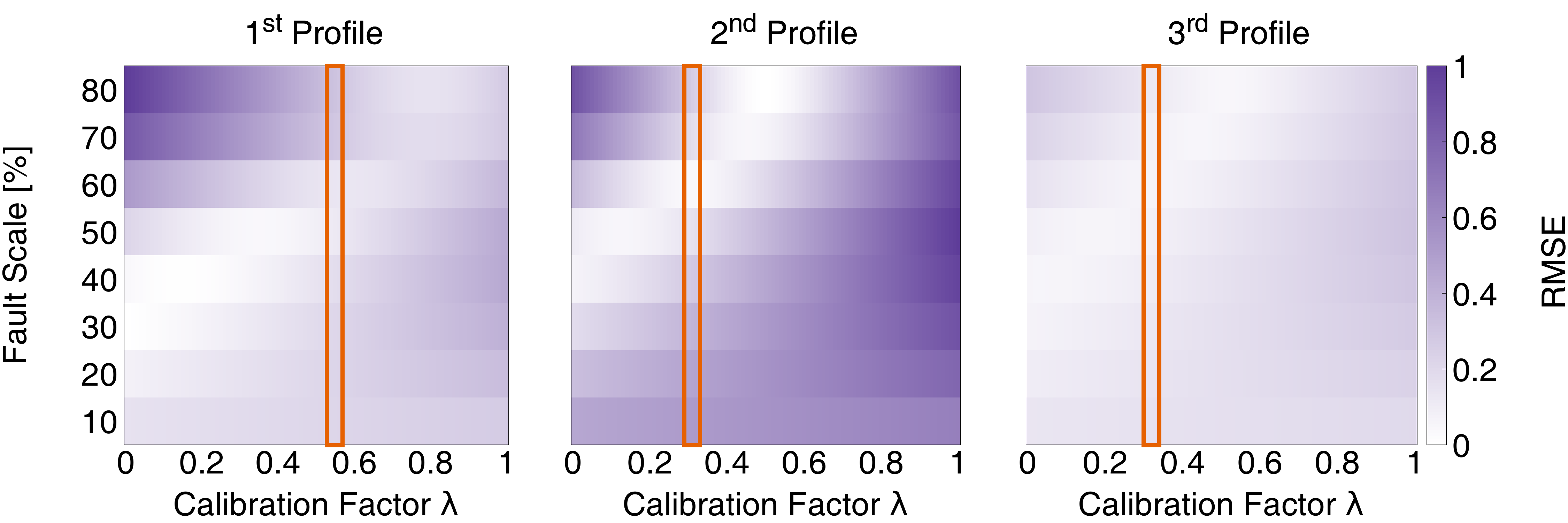}}	
\subfigure[]{\includegraphics[width=0.8\columnwidth]{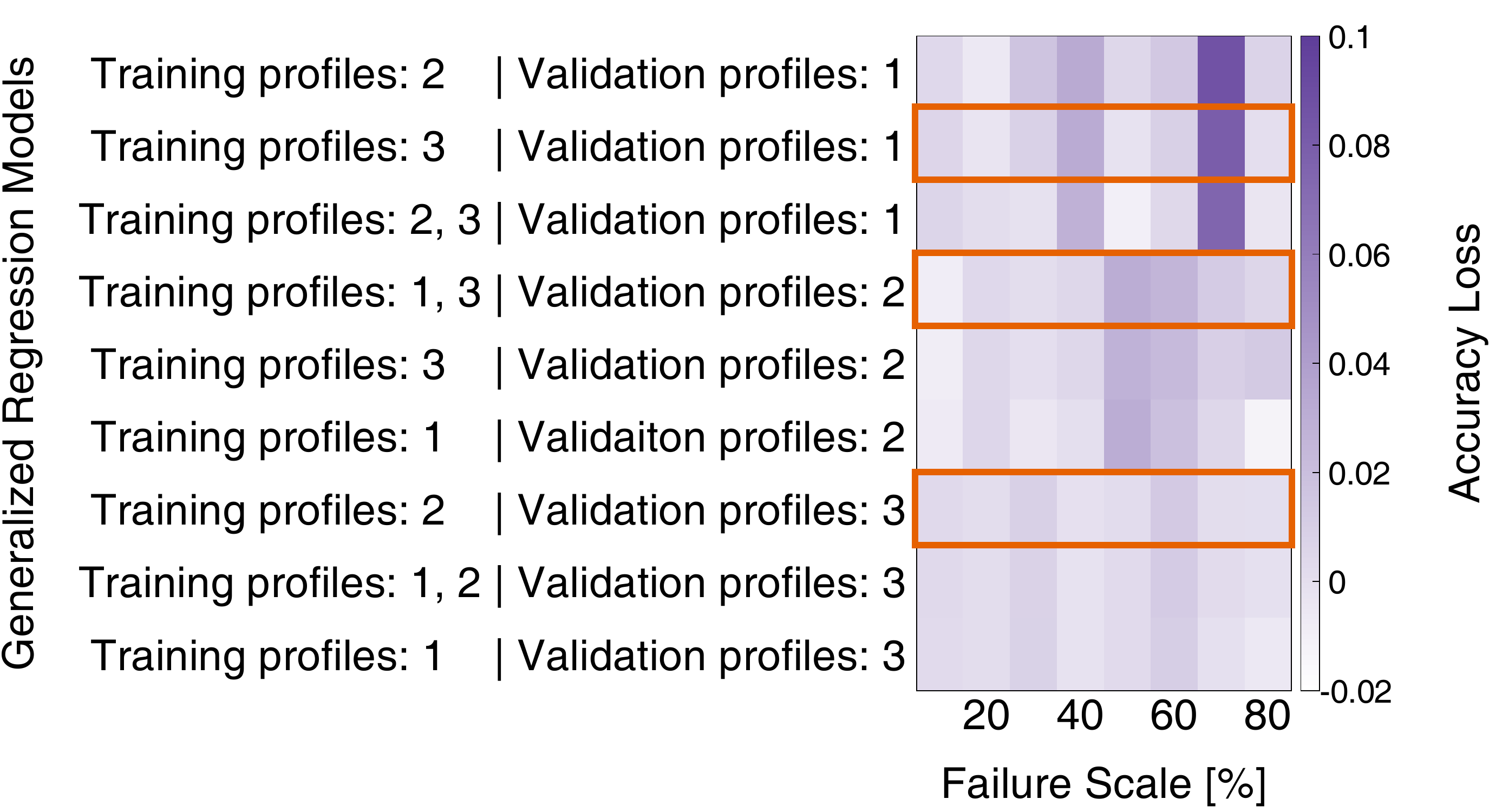}}	
\caption{Calibration configurations and their prediction performance for two calibration methods: (a) False negative calibration. (b) Generalized regression. The best calibration configurations are marked for further comparison of the different methods.}\label{fig:parameter-configuration}
\end{figure}

The other two model calibration methods are designed as follows: For each experimental setting, a feature vector of size $12 \times 5+2 = 62$ is constructed. This vector contains 5 quantiles (\nth{10}, \nth{30}, \nth{50}, \nth{70}, \nth{90}) of inconsistency cost for each of the $12$ calculations of the fault scenarios (Table~\ref{table:scenarios}). These values are extracted from the fault profiles applied to the peer sampling service~\cite{Jelasity2007}. The feature vector also contains the respective relative (to the maximum of 800) threshold and the fault scale for each experimental setting. All values of the feature vector are in the range $[0,1]$. Application-level data, i.e. the DIAS inconsistency cost, are used as target values for training, while features are agnostic of DIAS. Regression relies on the ordinary least squares model and its Python implementation of the statsmodels\footnote{Available at https://www.statsmodels.org/stable/index.html (last access: March 2021).}. The prediction based on linear regression is validated with two schemes: 

\subsubsection{Regression}\label{subsec:calibration-regression}

This second scheme uses all 696 experimental settings to train the linear regression model without regularization. It represents the best possible fit (intentional overfit) to the inconsistency costs observed in DIAS. 

\subsubsection{Generalized regression}\label{subsec:calibration-generalized-regression}

In this third scheme, training is limited to certain fault profiles and validation is performed on profiles on which training is not performed, assuming that the inconsistency cost for different fault profiles is generated from the same distribution. Figure~\ref{fig:parameter-configuration}b illustrates the prediction performance of all possible combinations of training and validation fault profiles for a generalized regression, measured with the accuracy loss: 

\vspace{-0.1cm}
\begin{equation}
RMSE(C_{\mathsf{GR}},C_{\mathsf{D}})-RMSE(C_{\mathsf{GR}},C_{\mathsf{R}}),
\end{equation}\label{eq:accuracy-loss}
\vspace{-0.3cm}

\noindent where $RMSE$ is the root mean square error, $C_{\mathsf{R}},C_{\mathsf{GR}}$ are the inconsistency cost of regression and generalized regression respectively and $C_{\mathsf{D}}$ is the predicted inconsistency cost of DIAS, i.e. the average relative approximation error of the summation. For each of the three validation fault profiles, the best fits observed in Figure~\ref{fig:parameter-configuration}b are selected to compare generalized regression to the other predictors. Generalized regression is performed with regularization\footnote{Elastic net is used with strength parameter of $\alpha=0.07$ and $L_{1}=0.05$ representing the preference of LASSO regularization over the RIDGE one.}.

\section{Experimental Evaluation}\label{sec:results}

This section illustrates the profiling of the inconsistency cost and how it can be used to improve the effectiveness of self-healing in decentralized data aggregation. It also shows a comparison of the calibration methods.

\subsection{Profiling of inconsistency cost}\label{subsec:profiling}

For the first evaluation objective, the inconsistency cost is profiled as follows: The density of the inconsistency cost and the relative frequency of each fault scenario are measured under varying fault scales, fault profiles and fault-detection thresholds. Due to space limitations, Figures~\ref{fig:cost-profiling-first-profile}-\ref{fig:cost-profiling-third-profile} focus on the fault scales of 20\%, 50\% and 80\% that depict the overall trend. 


\begin{figure*}[!htb]
\centering
\subfigure[Fault scale: 20\%]{\includegraphics[width=1.25\columnwidth]{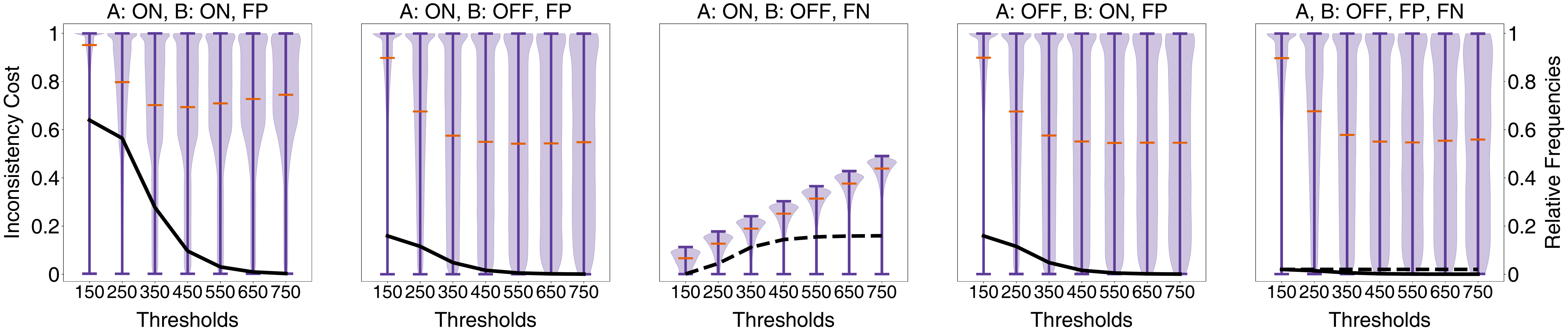}	}
\subfigure[Fault scale: 50\%]{\includegraphics[width=1.25\columnwidth]{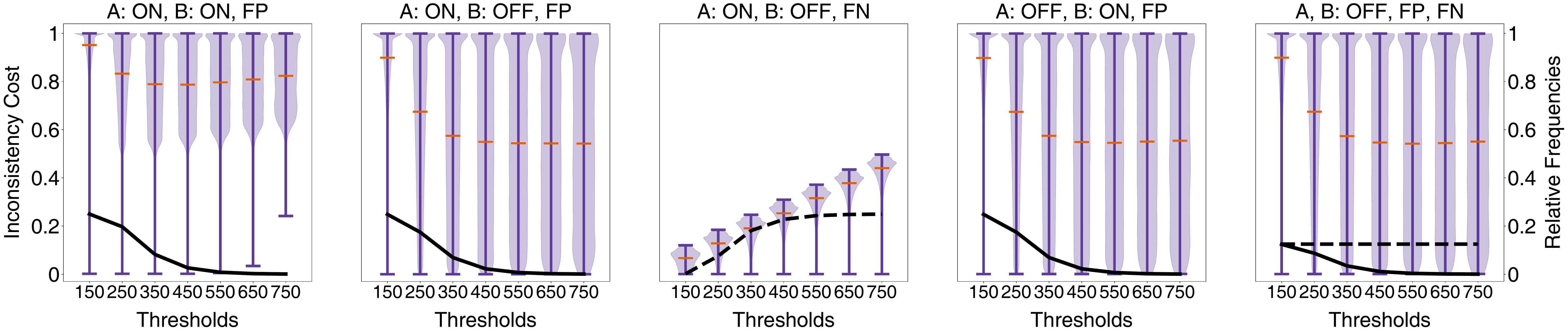}	}
\subfigure[Fault scale: 80\%]{\includegraphics[width=1.25\columnwidth]{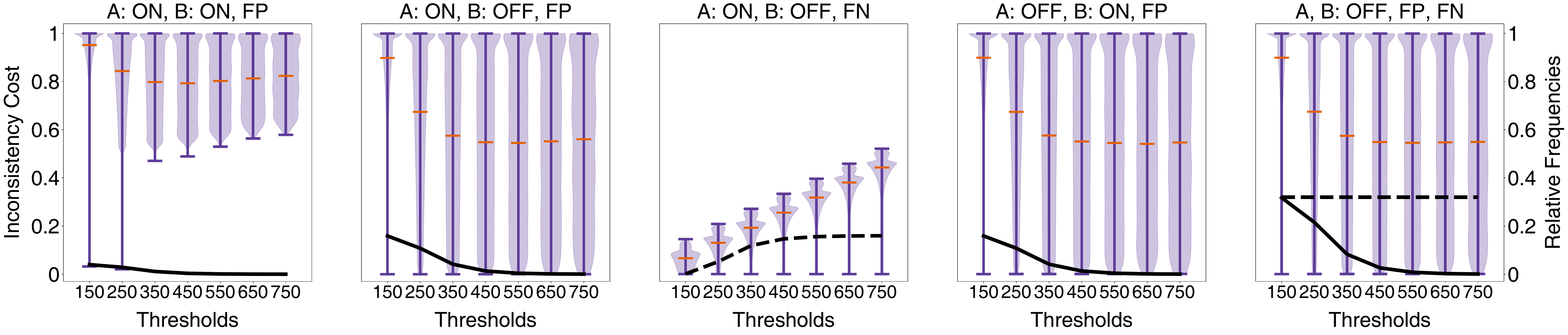}}
\caption{The inconsistency cost  of the fault scenarios (violins with density values on the left Y-axis) and their relative frequency (lines with values on the right Y-axis) under a fault scale of 20\%, 50\% and 80\% in the \nth{1} fault profile. The solid lines depict the relative frequency for false positive (FP) states, while the dashed lines the one for false negative (FN) states. }\label{fig:cost-profiling-first-profile}
\end{figure*}

\begin{figure*}[!htb]
\centering
\subfigure[Fault scale: 20\%]{\includegraphics[width=2.0\columnwidth]{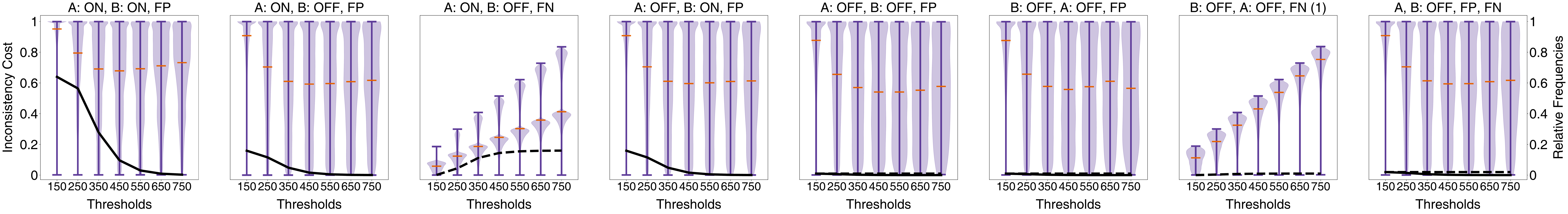}	}
\subfigure[Fault scale: 50\%]{\includegraphics[width=2.0\columnwidth]{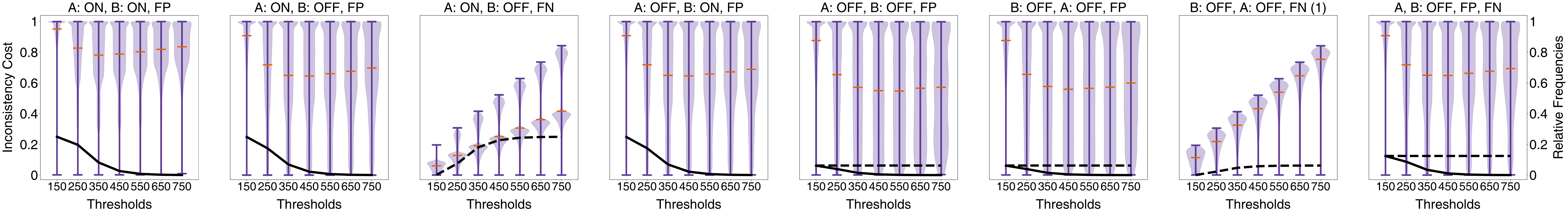}	}
\subfigure[Fault scale: 80\%]{\includegraphics[width=2.0\columnwidth]{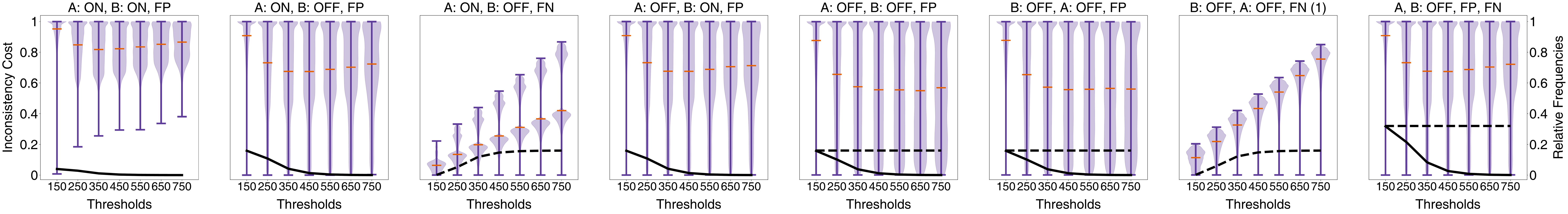}}
\caption{The inconsistency cost  of the fault scenarios (violins with density values on the left Y-axis) and their relative frequency (lines with values on the right Y-axis) under a fault scale of 20\%, 50\% and 80\% in the \nth{2} fault profile. The solid lines depict the relative frequency for false positive (FP) states, while the dashed lines the one for false negative (FN) states.}\label{fig:cost-profiling-second-profile}
\end{figure*}

\begin{figure*}[!htb]
\centering
\subfigure[Fault scale: 20\%]{\includegraphics[width=2.0\columnwidth]{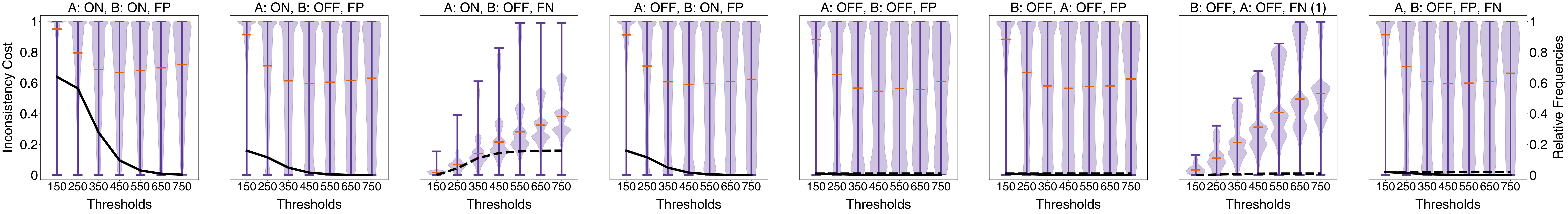}	}
\subfigure[Fault scale: 50\%]{\includegraphics[width=2.0\columnwidth]{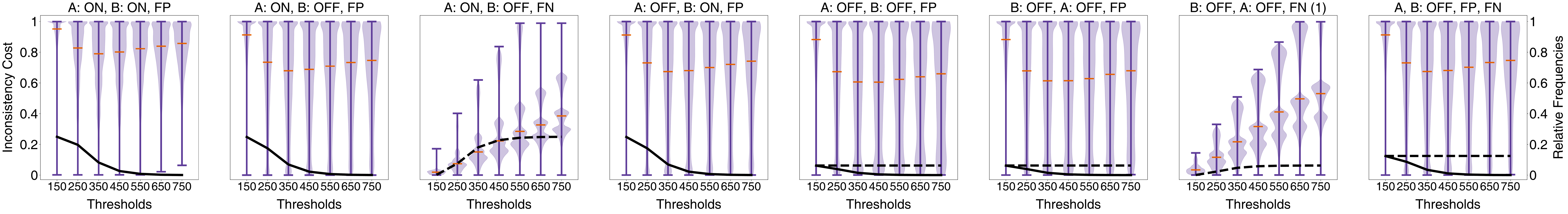}	}
\subfigure[Fault scale: 80\%]{\includegraphics[width=2.0\columnwidth]{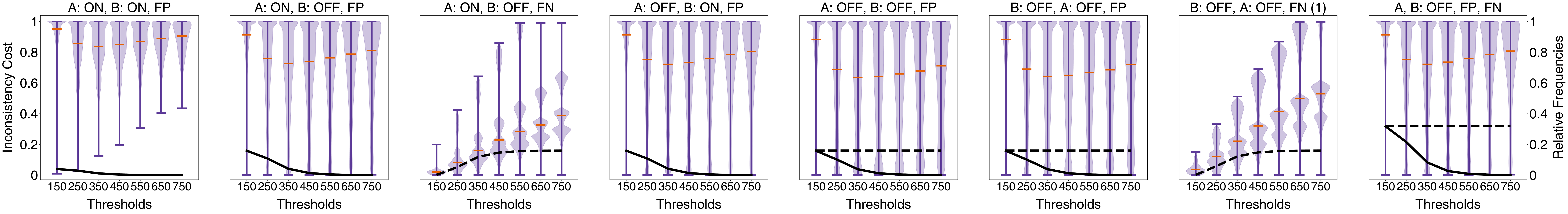}}
\caption{The inconsistency cost  of the fault scenarios (violins with density values on the left Y-axis) and their relative frequency (lines with values on the right Y-axis) under a fault scale of 20\%, 50\% and 80\% in the \nth{3} fault profile. The solid lines depict the relative frequency for false positive (FP) states, while the dashed lines the one for false negative (FP) states.}\label{fig:cost-profiling-third-profile}
\end{figure*}

The following observations can be made in Figures~\ref{fig:cost-profiling-first-profile}-\ref{fig:cost-profiling-third-profile}: (i) The inconsistency cost by false positive states has on average higher magnitude than the one of false negative states across the different fault profiles and  scales. (ii) With an increasing fault scale, the inconsistency cost slightly increases, especially for the fault-scenario $A$: ON, $B$: ON, false positive.  However, the relative frequency of the inconsistency cost for this fault scenario decreases, in exchange for an increase in fault scenarios $A$: OFF, B: OFF, false positive, $B$: OFF, $A$: OFF, false positive, $B$: OFF, $A$: OFF, false negative and $A$, $B$: OFF, false positive and negative. (iii) For each fault profile when nodes do not fail, the magnitude of the inconsistency cost by false positives is respectively 22.37\%, 17.2\% and 16.07\% higher on average than the one with defecting nodes. For fault scales of 20\%, 50\% and 80\%, this difference is 13.0\%, 20.99\% and 22.63\% higher on average when nodes do not fail. (iv) The inconsistency cost by false positives is minimized for middle threshold values, i.e. 450 epochs for 20\% fault scale, 350 epochs for 50\% and 80\% fault scale. There thresholds though depend on the system size and parameters with which the peer sampling service is chosen to operate, i.e. partial view size, execution period, swap/healer parameters~\cite{Jelasity2007,Pournaras2017}. In other words, different configurations of the underlying system yield to a different profiling of the inconsistency cost. (v) In the second profile, the density of the inconsistency cost for the fault scenario of $A$: ON, $B$: OFF, false negative, has two peaks that originate from the two different times in which the nodes defect (respectively three peaks at the 3rd profile). Larger thresholds shift the peaks to larger inconsistency costs ($d-F_{B}$ is maximized) and increase the distance between the peaks as also confirmed for the fault scenario $B$: OFF, $A$: OFF, false negative. (vi) The relative frequency of fault scenarios with a false positive state decreases for higher thresholds, while it increases or remains constant for a false negative state. All these observations confirm that the profiling of the inconsistency cost generated by the fault scenarios can provide a highly insightful analysis of the trade-offs involved in tuning fault-detection mechanisms in decentralized systems with uncertainties.

\subsection{Self-healing decentralized data aggregation}\label{subsec:profiling}

More cost-effective self-healing mechanisms can be designed, tailored to minimize the predicted inconsistency cost of specific fault scenarios. Note for instance Figure~\ref{fig:aggregation} that illustrates the applicability of self-healing in DIAS in the three fault profiles and the fault scales of 20\%, 50\% and 80\%. The actual aggregate of summation is compared to the faulty estimate (no corrective operations) and two corrective estimates (without any calibration): (i) A reference with a fixed threshold at 100 epochs. (ii) The one with the threshold that minimizes the inconsistency cost. Therefore, the profiling of the inconsistency cost provides the required tuning to fault detection to minimize the relative approximation error of the aggregation. The root mean square error between the actual sum and the faulty estimate (no self-corrective operations) is on average 28.17\% higher than the DIAS estimate with the threshold resulting in minimal inconsistency cost. Across fault profiles, the corresponding errors are 4.29\%, 34.83\% and 34.72\% higher for the fault scales of 20\%, 50\% and 80\%, respectively. The DIAS corrective estimate with reference threshold $t=100$ performs worse than the faulty estimate across all fault profiles and scales\footnote{On average, the root mean square error between the actual sum and the faulty estimate is 113.07\% lower than the DIAS estimate with a reference threshold of $t=100$}, demonstrating the implication of an erroneous fault correction and, apparently, how dramatic can a misconfiguration of fault detection be for a decentralized application.

\begin{figure}[!htb]
\centering
\subfigure[\nth{1} fault profile, 20\% fault scale]{\includegraphics[width=0.31\columnwidth]{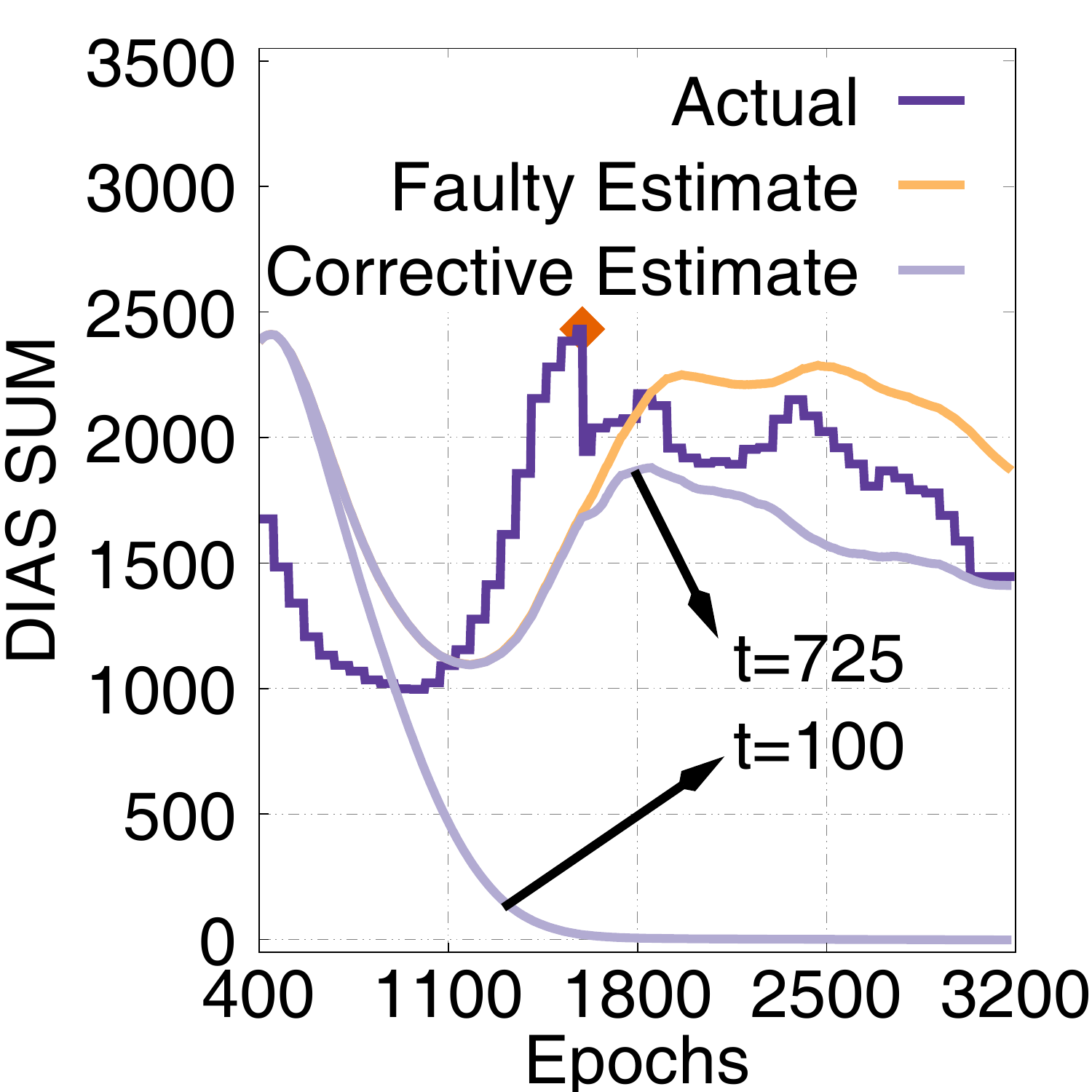}}
\subfigure[\nth{1} fault profile, 50\% fault scale]{\includegraphics[width=0.31\columnwidth]{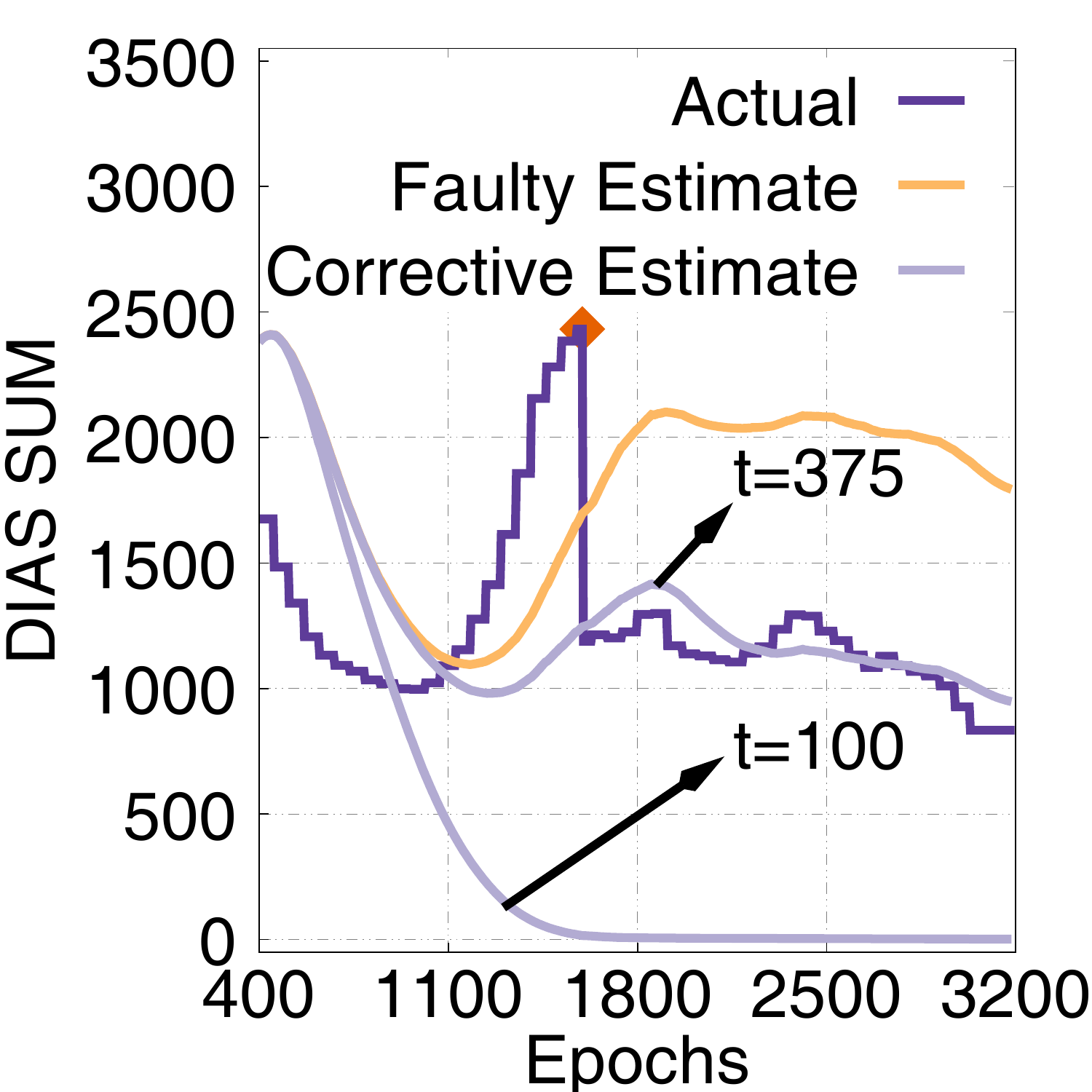}}
\subfigure[\nth{1} fault profile, 80\% fault scale]{\includegraphics[width=0.31\columnwidth]{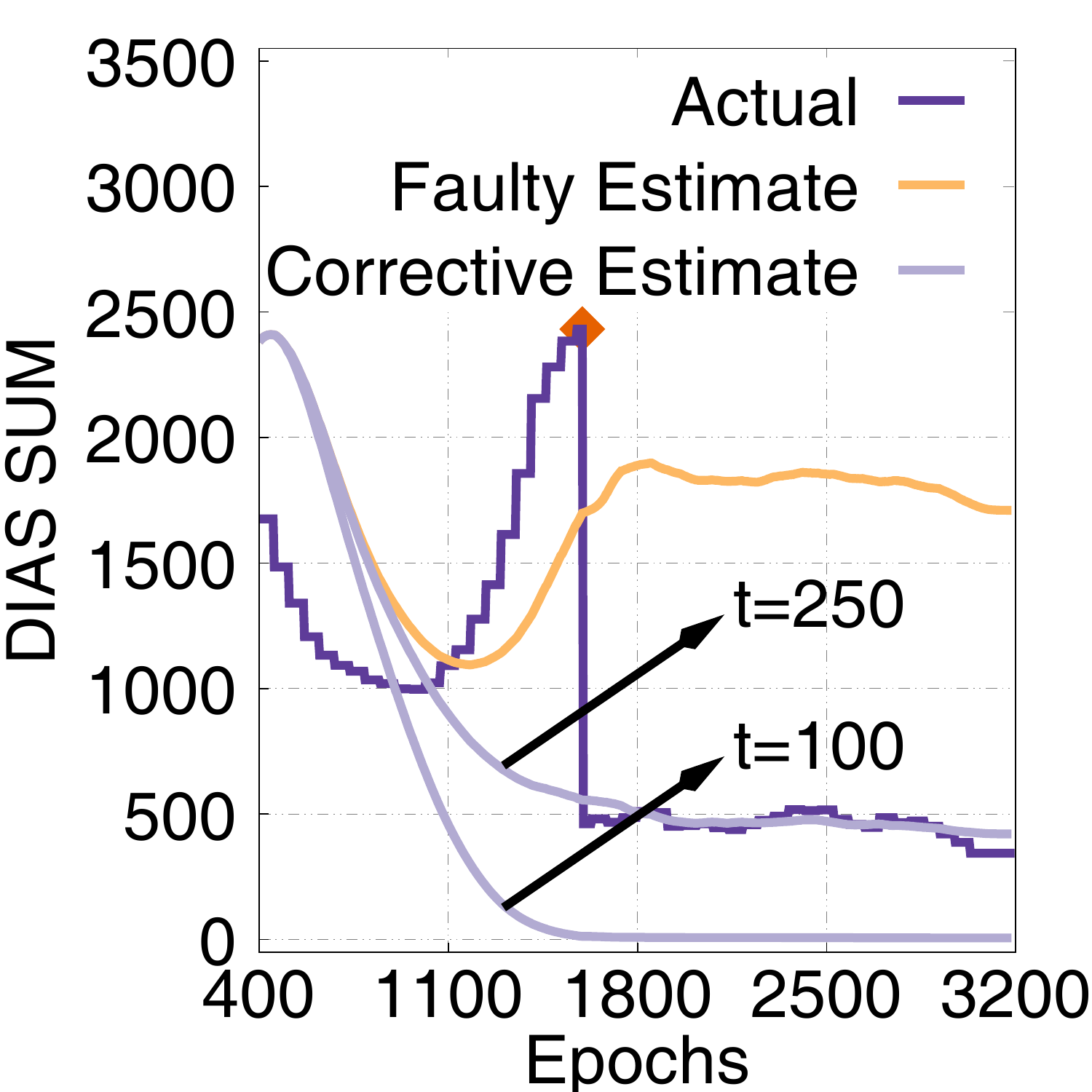}} 
\subfigure[\nth{2} fault profile, 20\% fault scale]{\includegraphics[width=0.31\columnwidth]{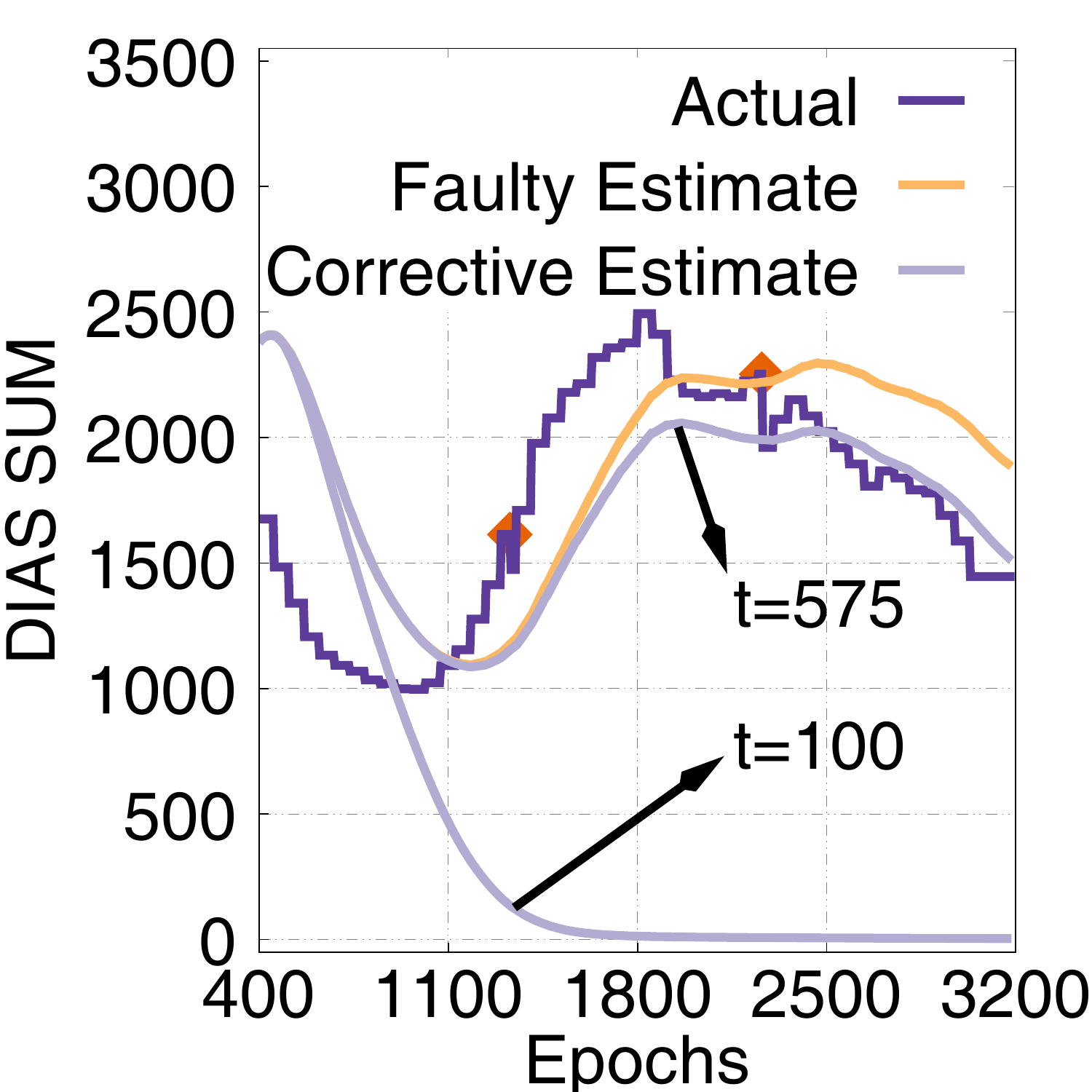}}
\subfigure[\nth{2} fault profile, 50\% fault scale]{\includegraphics[width=0.31\columnwidth]{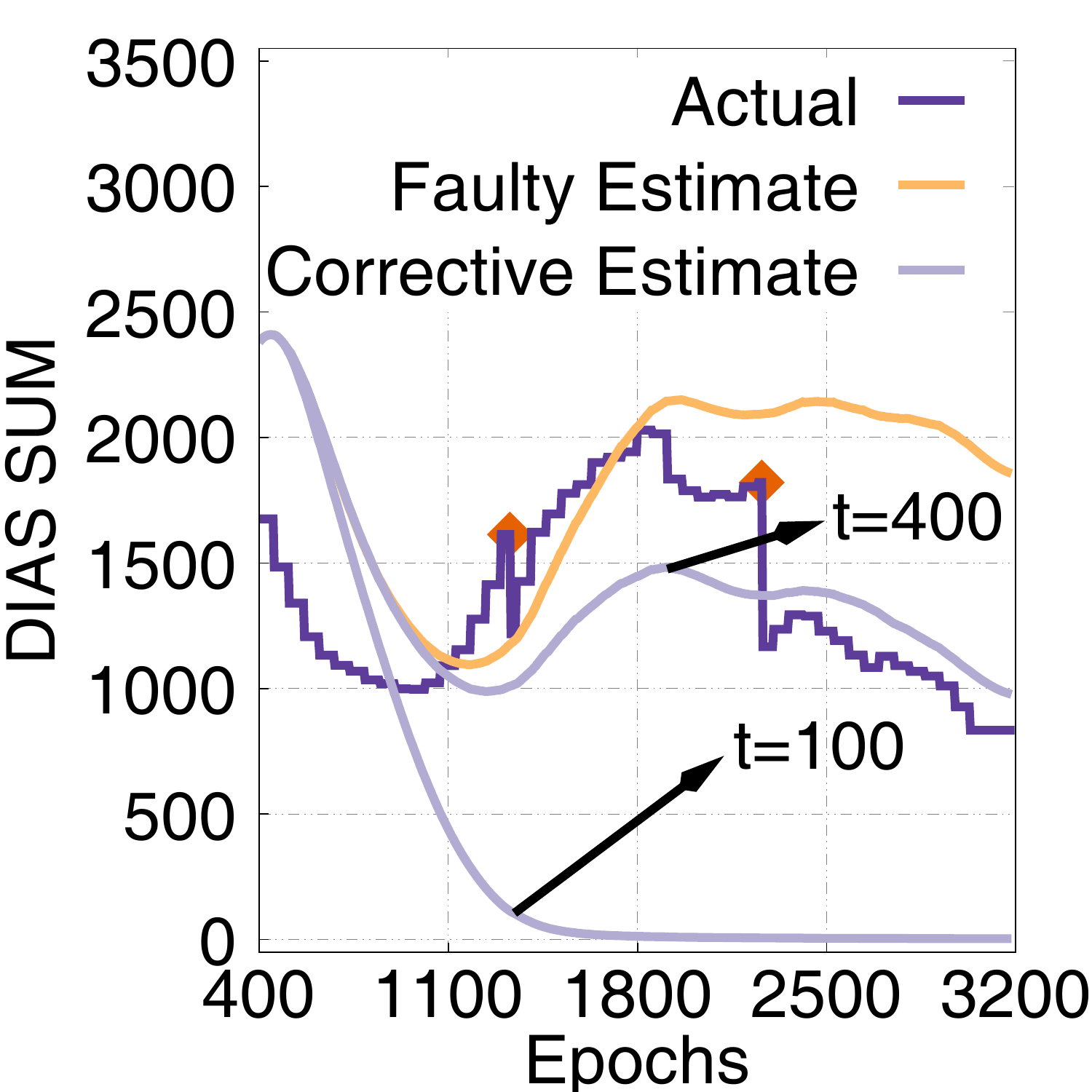}}
\subfigure[\nth{2} fault profile, 80\% fault scale]{\includegraphics[width=0.31\columnwidth]{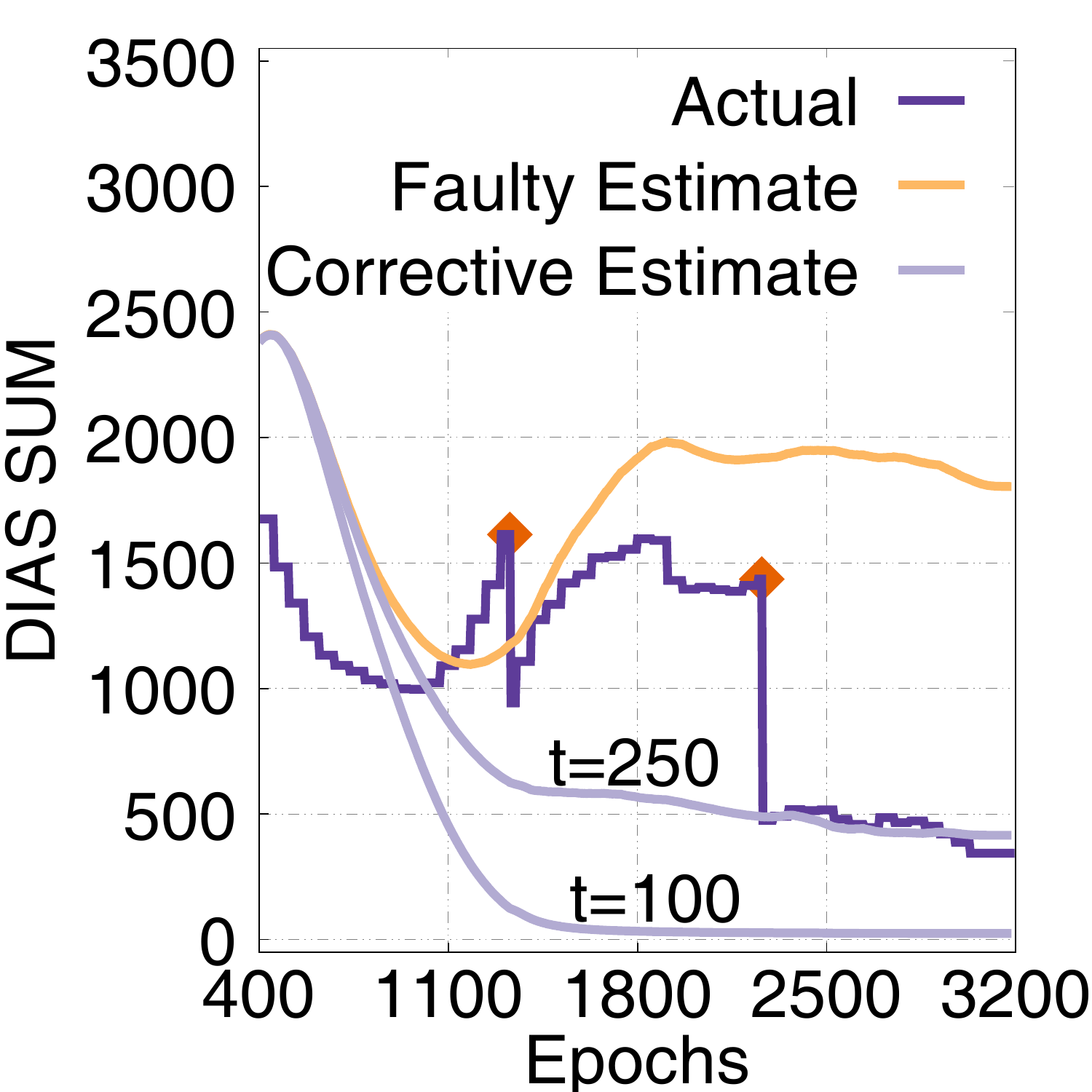}} 
\subfigure[\nth{3} fault profile, 20\% fault scale]{\includegraphics[width=0.31\columnwidth]{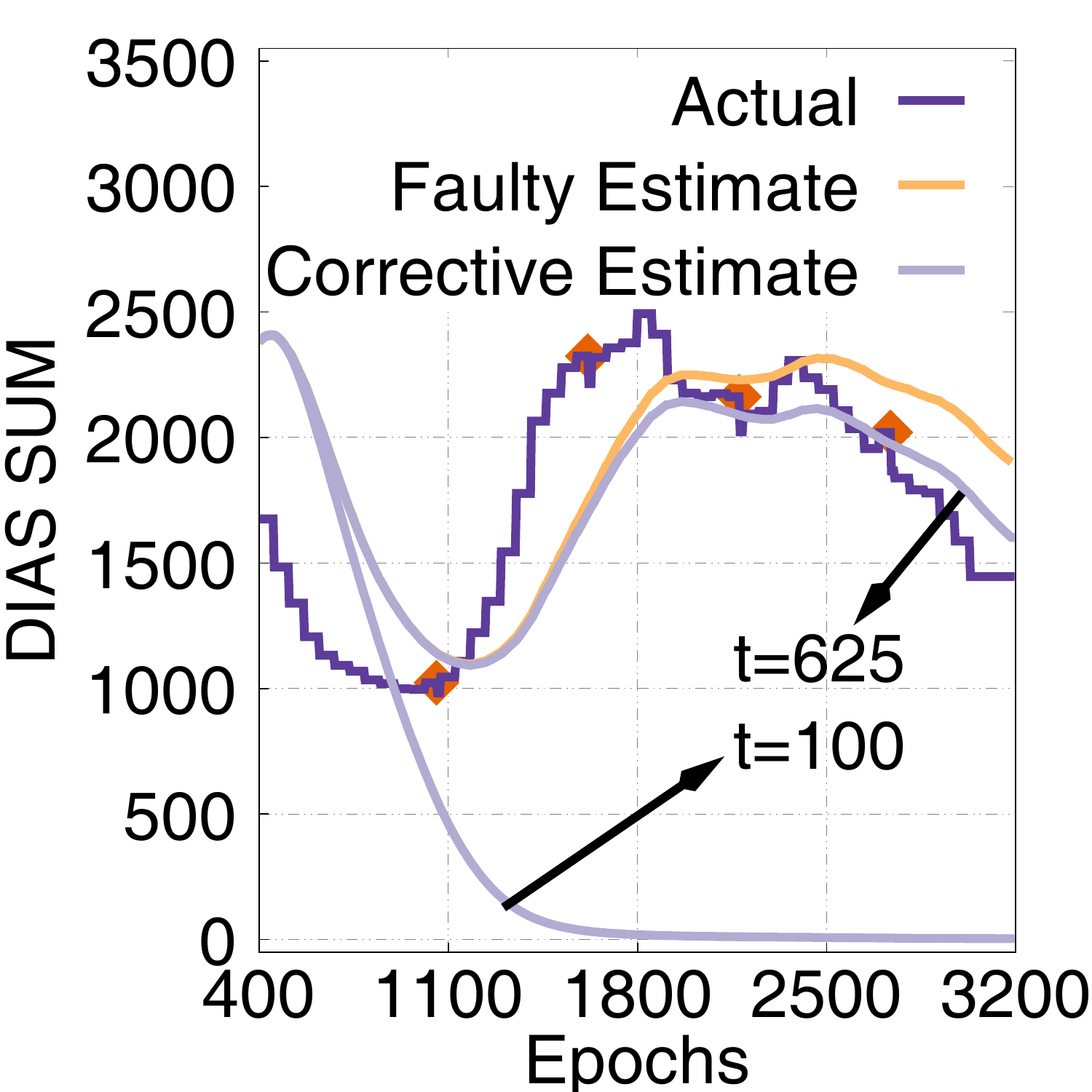}}
\subfigure[\nth{3} fault profile, 50\% fault scale]{\includegraphics[width=0.31\columnwidth]{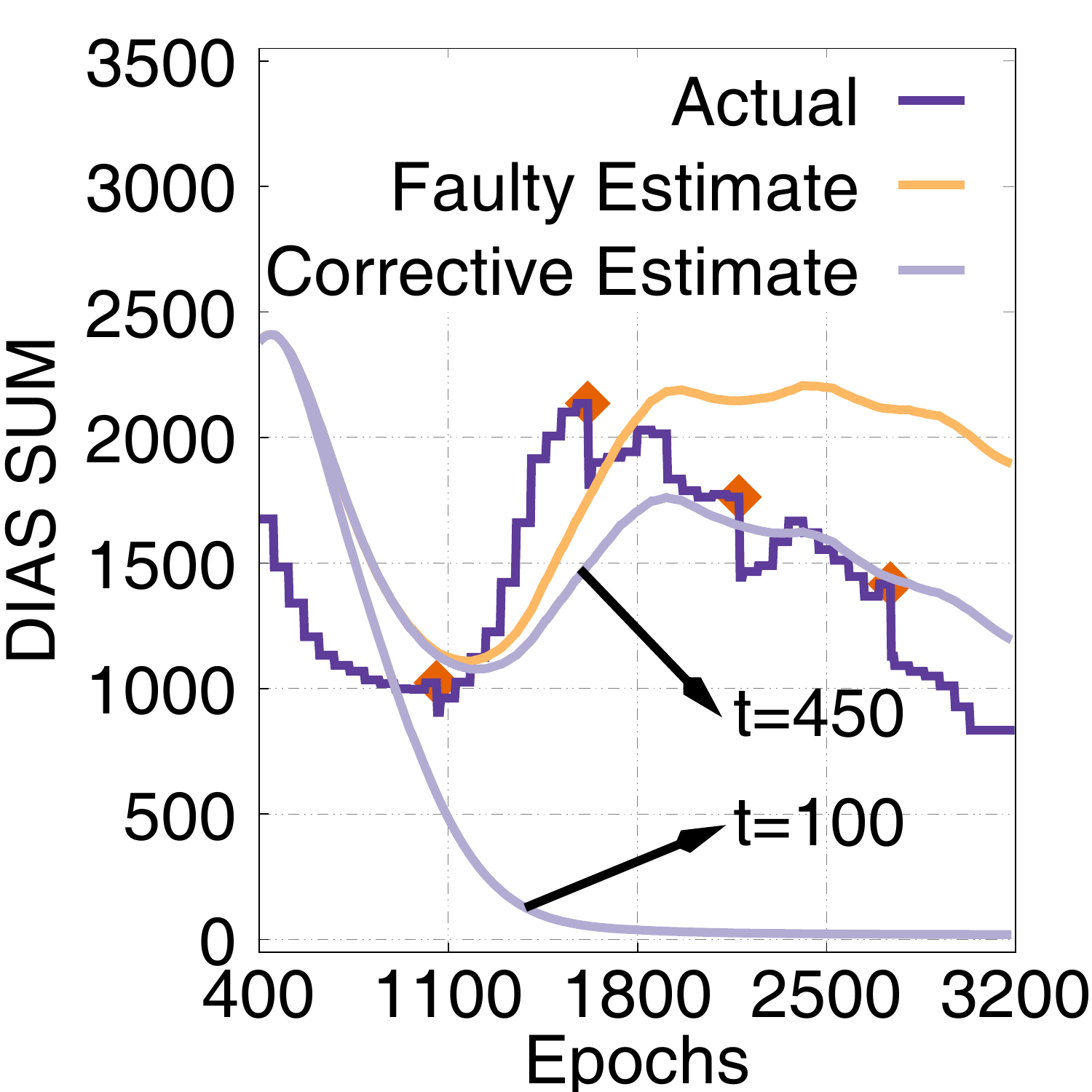}}
\subfigure[\nth{3} fault profile, 80\% fault scale]{\includegraphics[width=0.31\columnwidth]{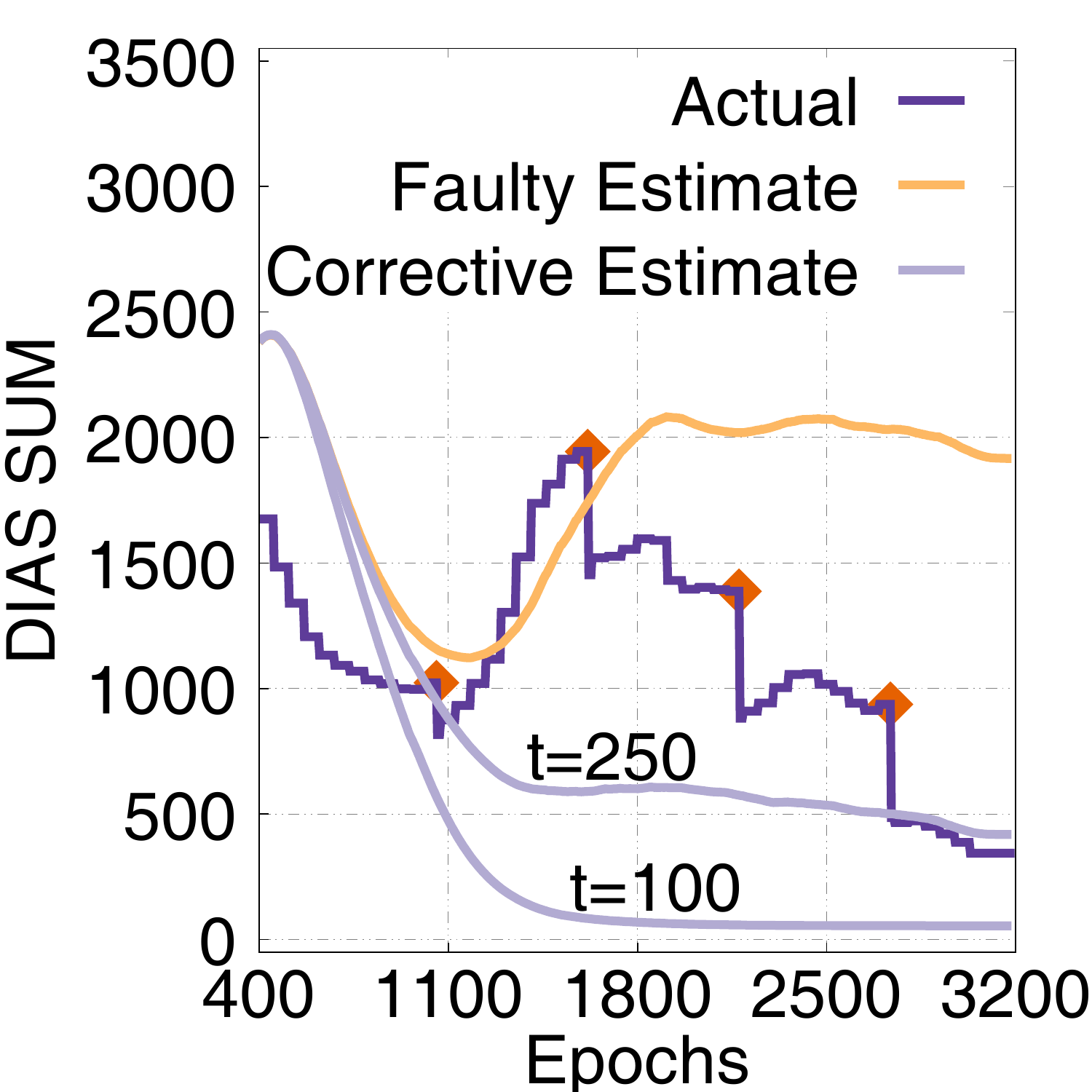}} 
\caption{DIAS self-healing under 20\%, 50\% and 80\% fault scales in the three fault profiles. In the corrective estimates, the threshold of $t=100$ is shown as reference vs. the threshold that minimizes the inconsistency cost.}\label{fig:aggregation}
\end{figure}

\subsection{Evaluation of model calibration methods}\label{subsec:profiling}

Figure~\ref{fig:predictions} addresses the second and third objective of the experimental evaluation that is the predictive performance of the inconsistency cost by the modeled fault scenarios. The average relative approximation error of the DIAS sum estimations is compared to the calibrated predictions made by the modeled fault scenarios under different thresholds in the three fault profiles. 

\begin{figure}[!htb]
\centering
\subfigure[\nth{1} fault profile, 20\% fault scale]{\includegraphics[width=0.31\columnwidth]{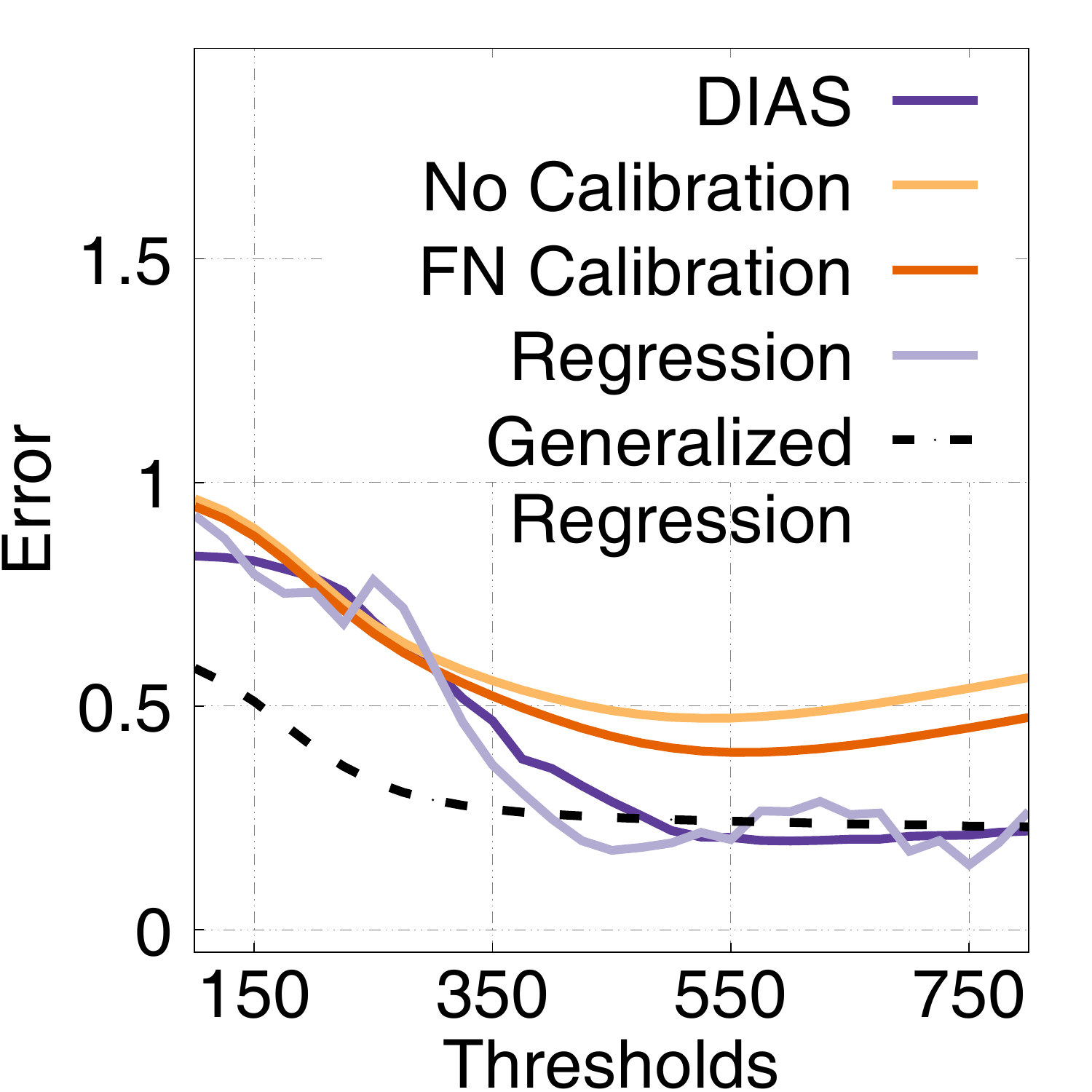}}
\subfigure[\nth{1} fault profile, 50\% fault scale]{\includegraphics[width=0.31\columnwidth]{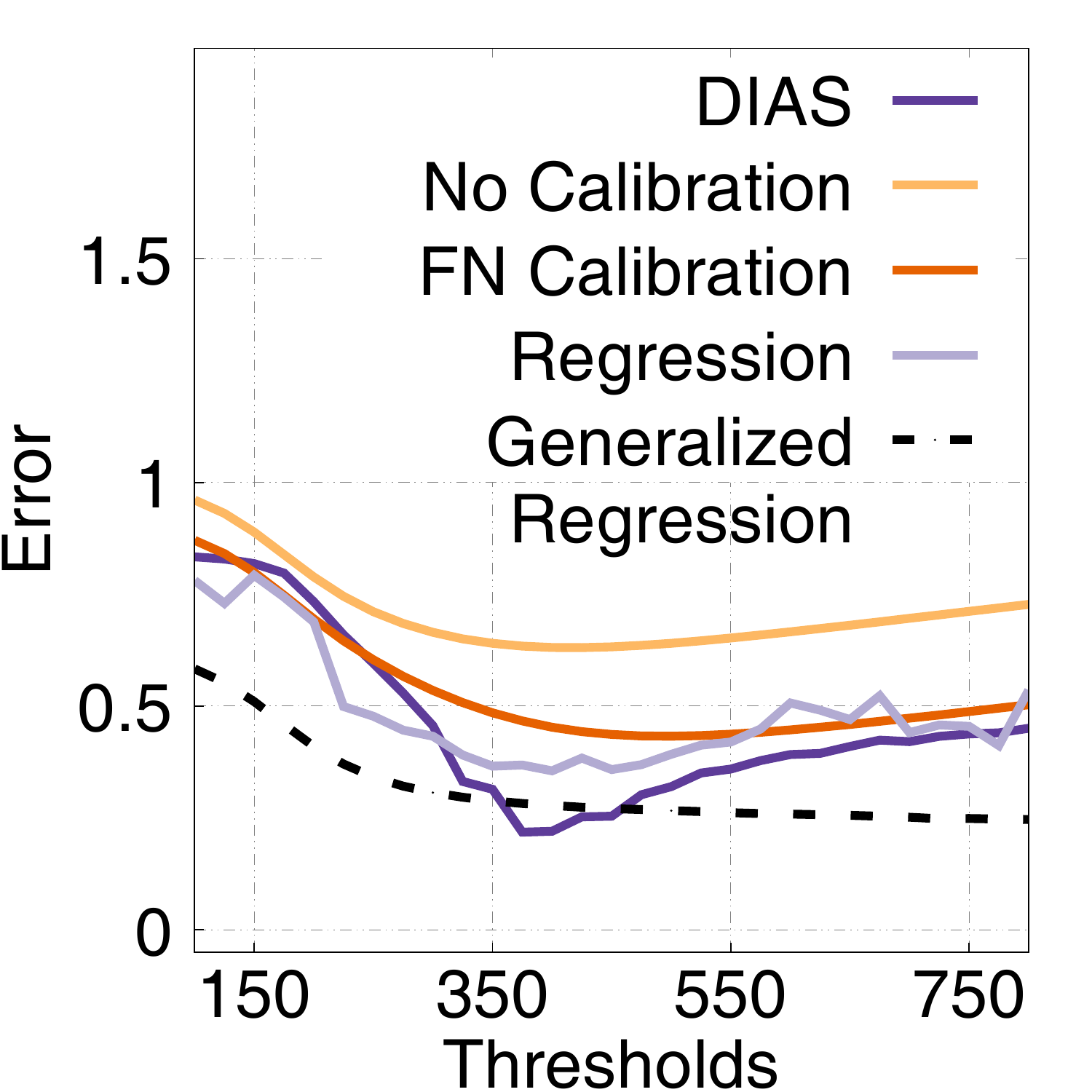}}
\subfigure[\nth{1} fault profile, 80\% fault scale]{\includegraphics[width=0.31\columnwidth]{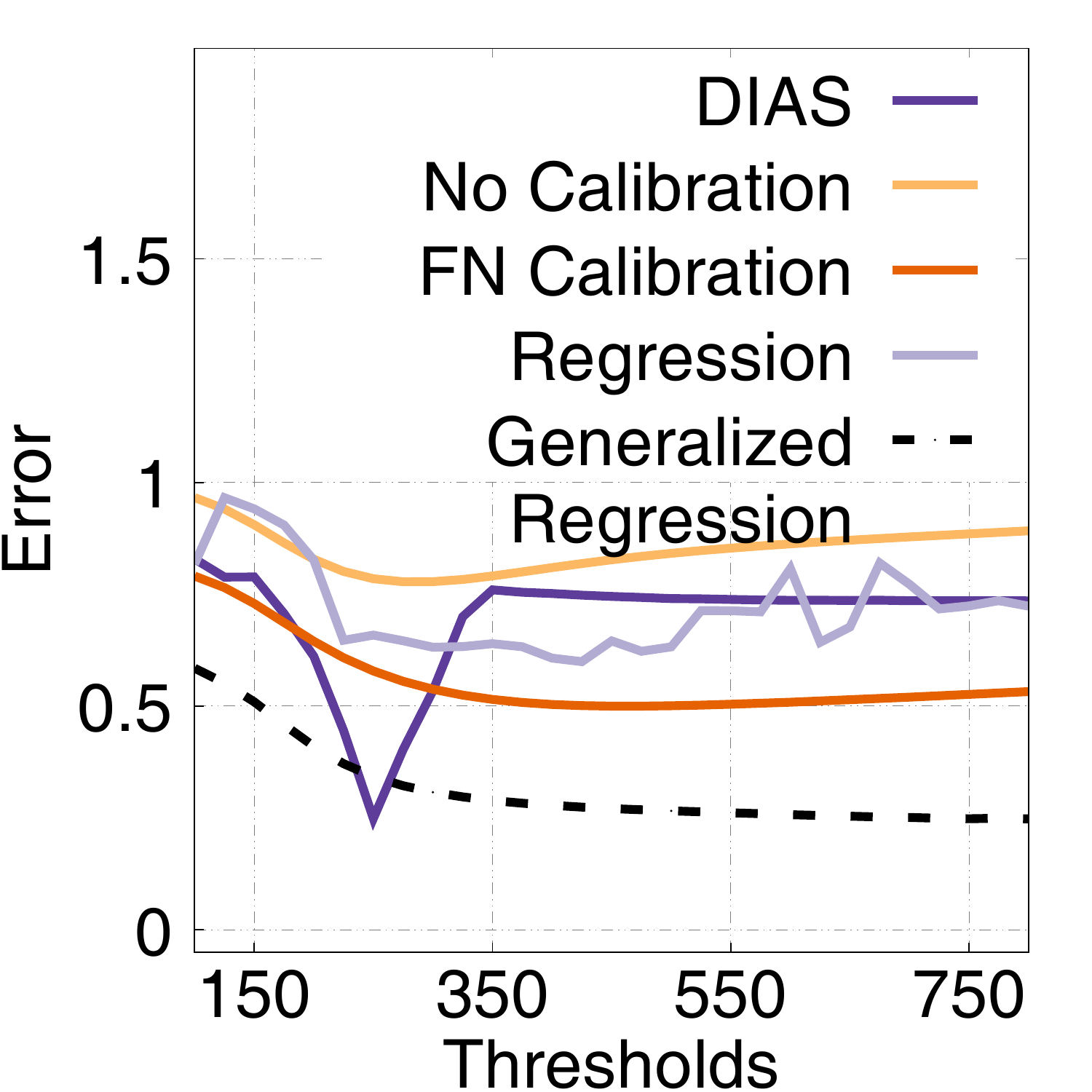}}	
\subfigure[\nth{2} fault profile, 20\% fault scale]{\includegraphics[width=0.31\columnwidth]{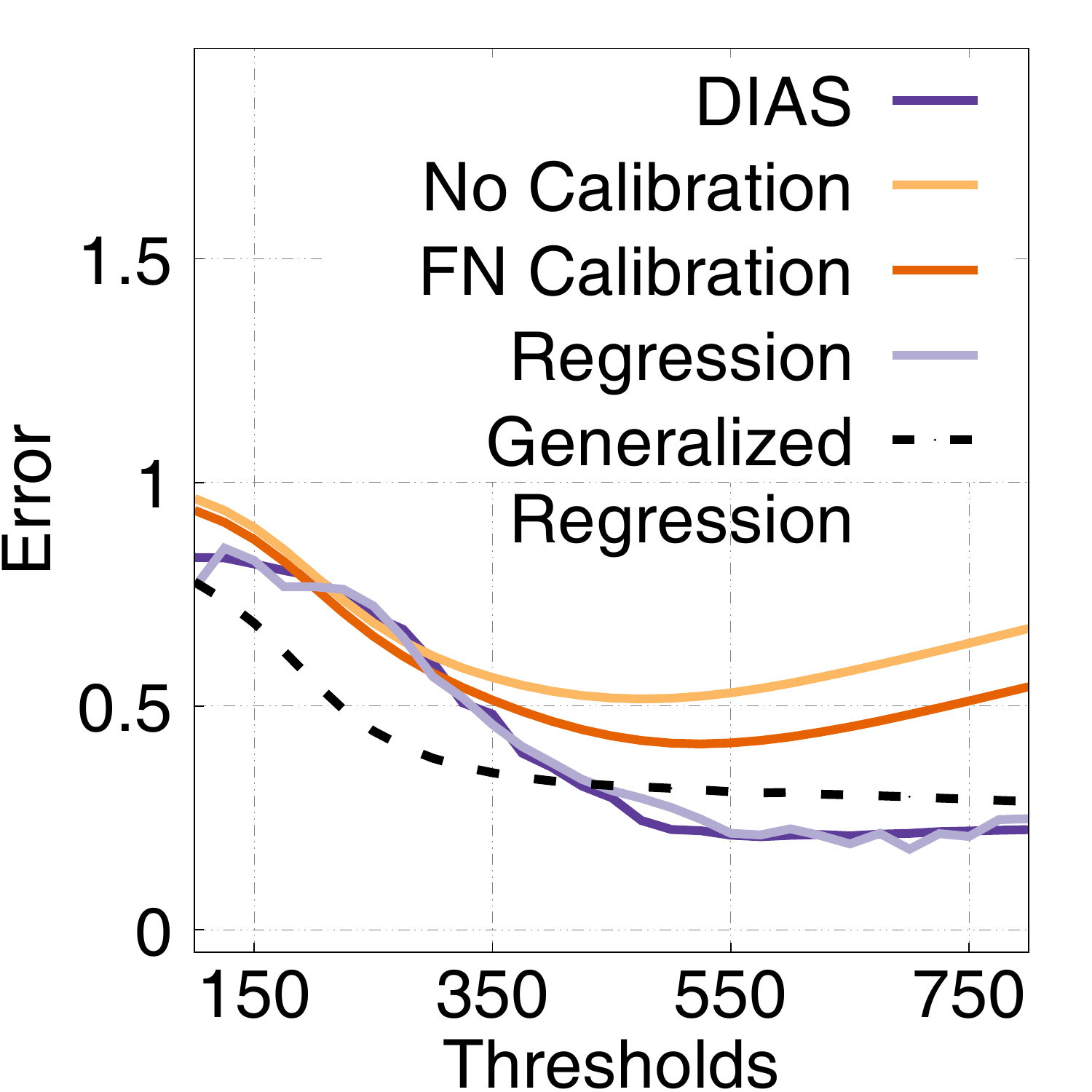}}
\subfigure[\nth{2} fault profile, 50\% fault scale]{\includegraphics[width=0.31\columnwidth]{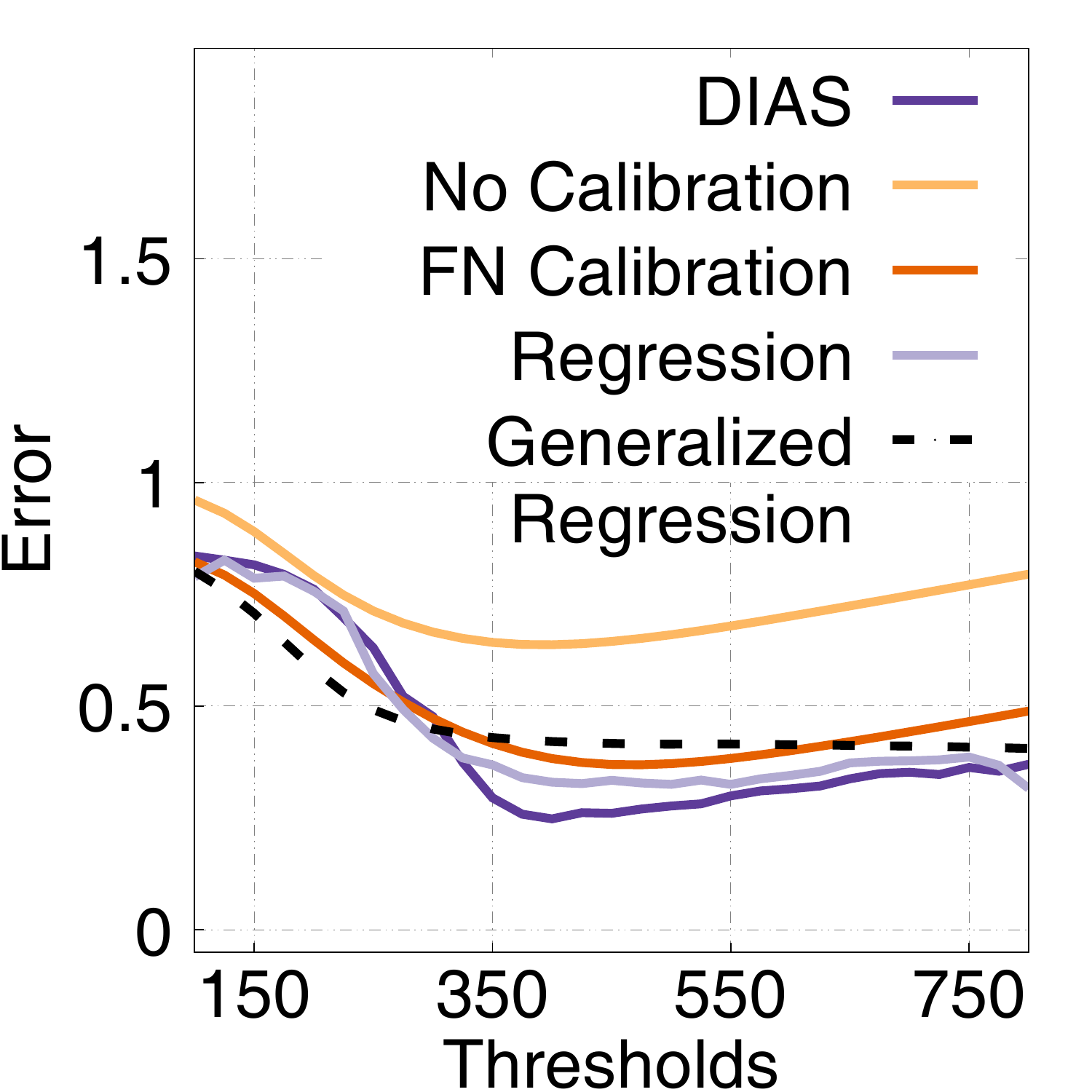}}
\subfigure[\nth{2} fault profile, 80\% fault scale]{\includegraphics[width=0.31\columnwidth]{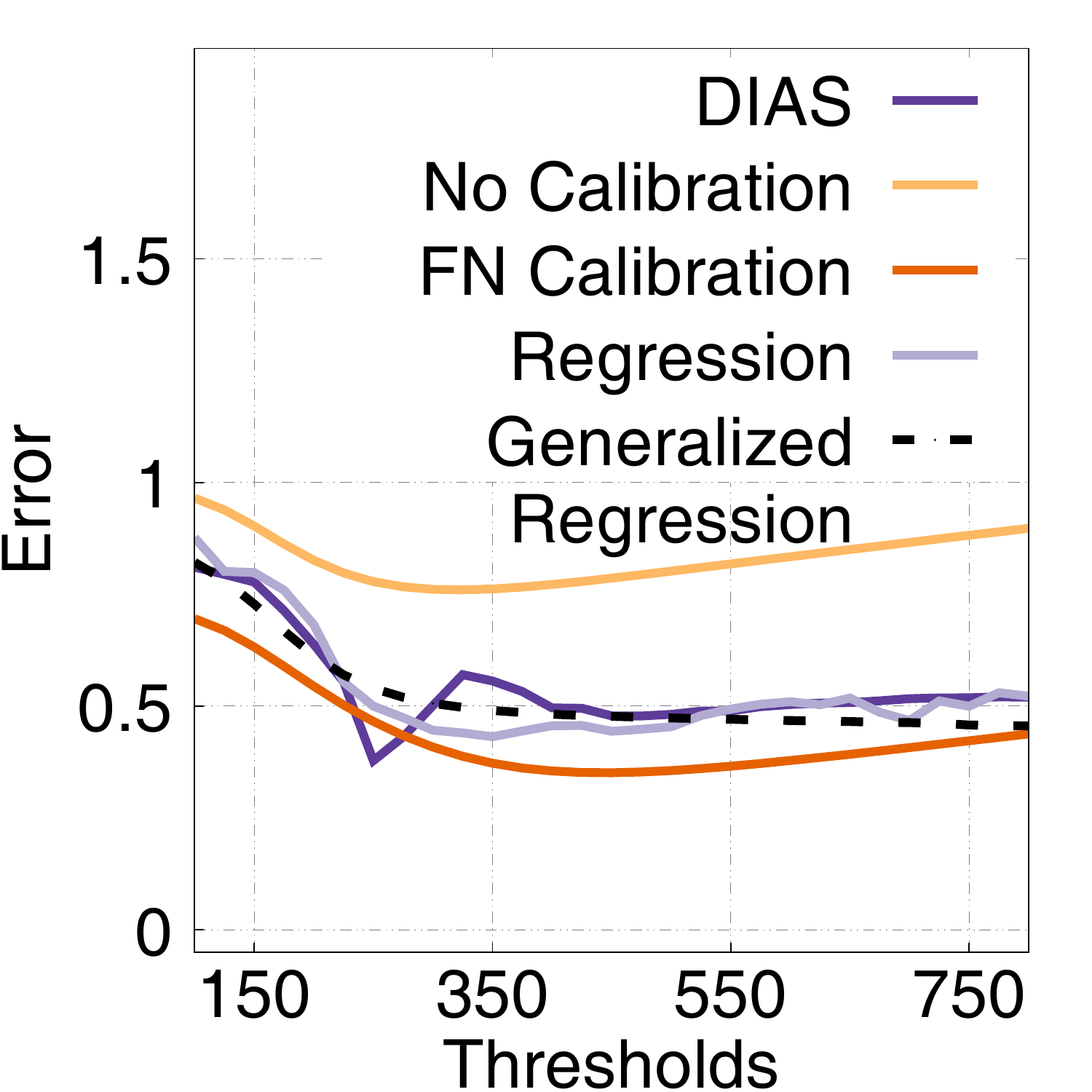}}	
\subfigure[\nth{3} fault profile, 20\% fault scale]{\includegraphics[width=0.31\columnwidth]{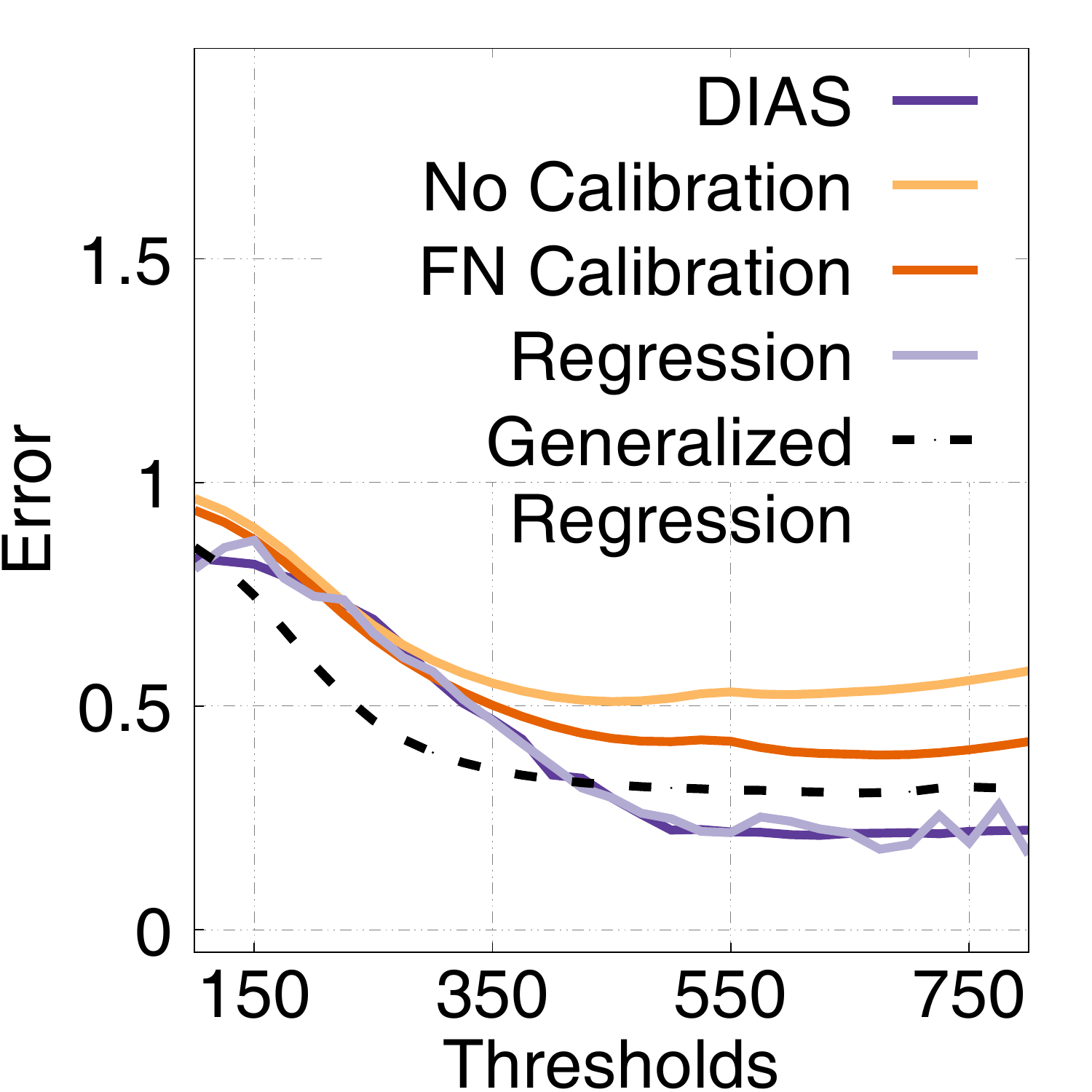}}
\subfigure[\nth{3} fault profile, 50\% fault scale]{\includegraphics[width=0.31\columnwidth]{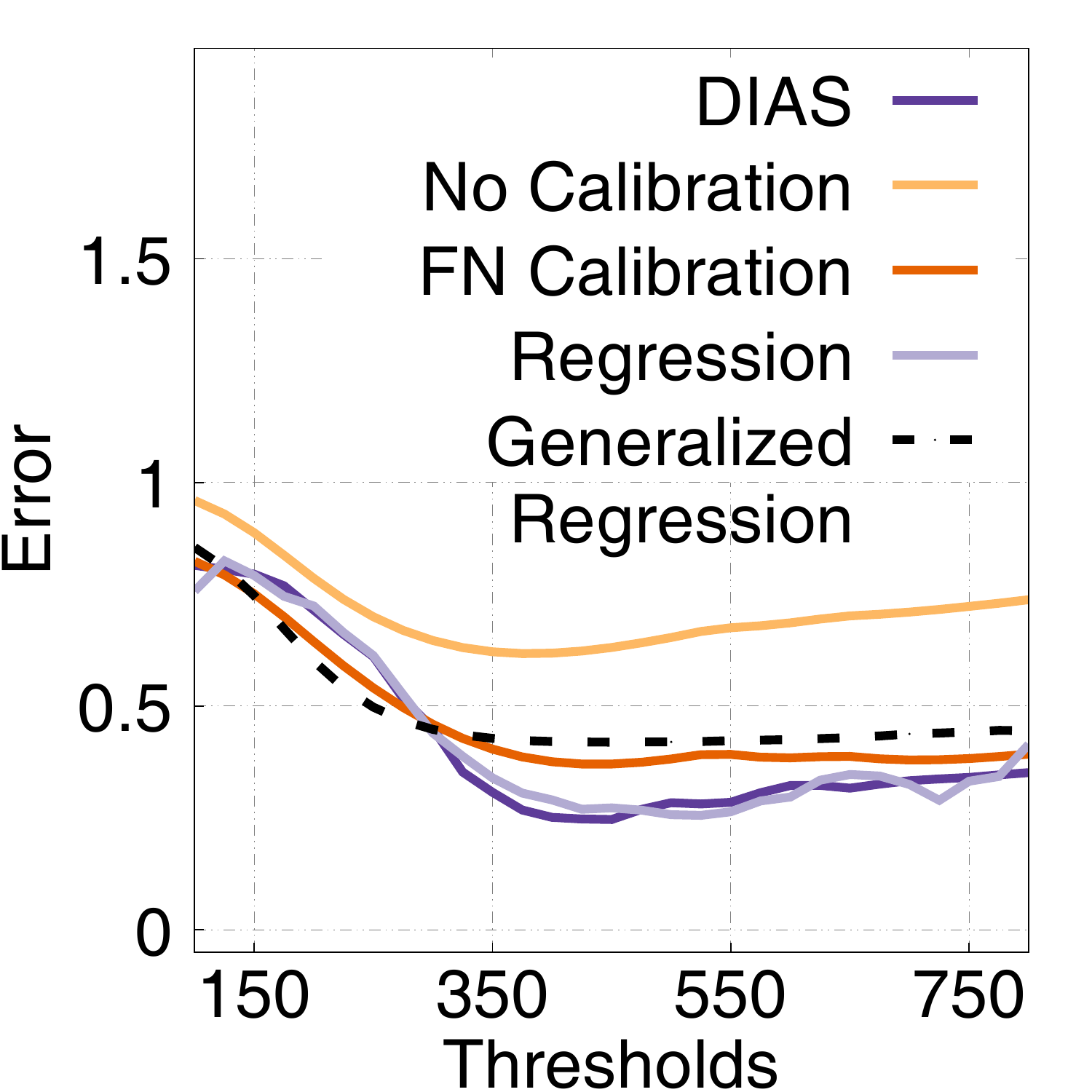}}
\subfigure[\nth{3} fault profile, 80\% fault scale]{\includegraphics[width=0.31\columnwidth]{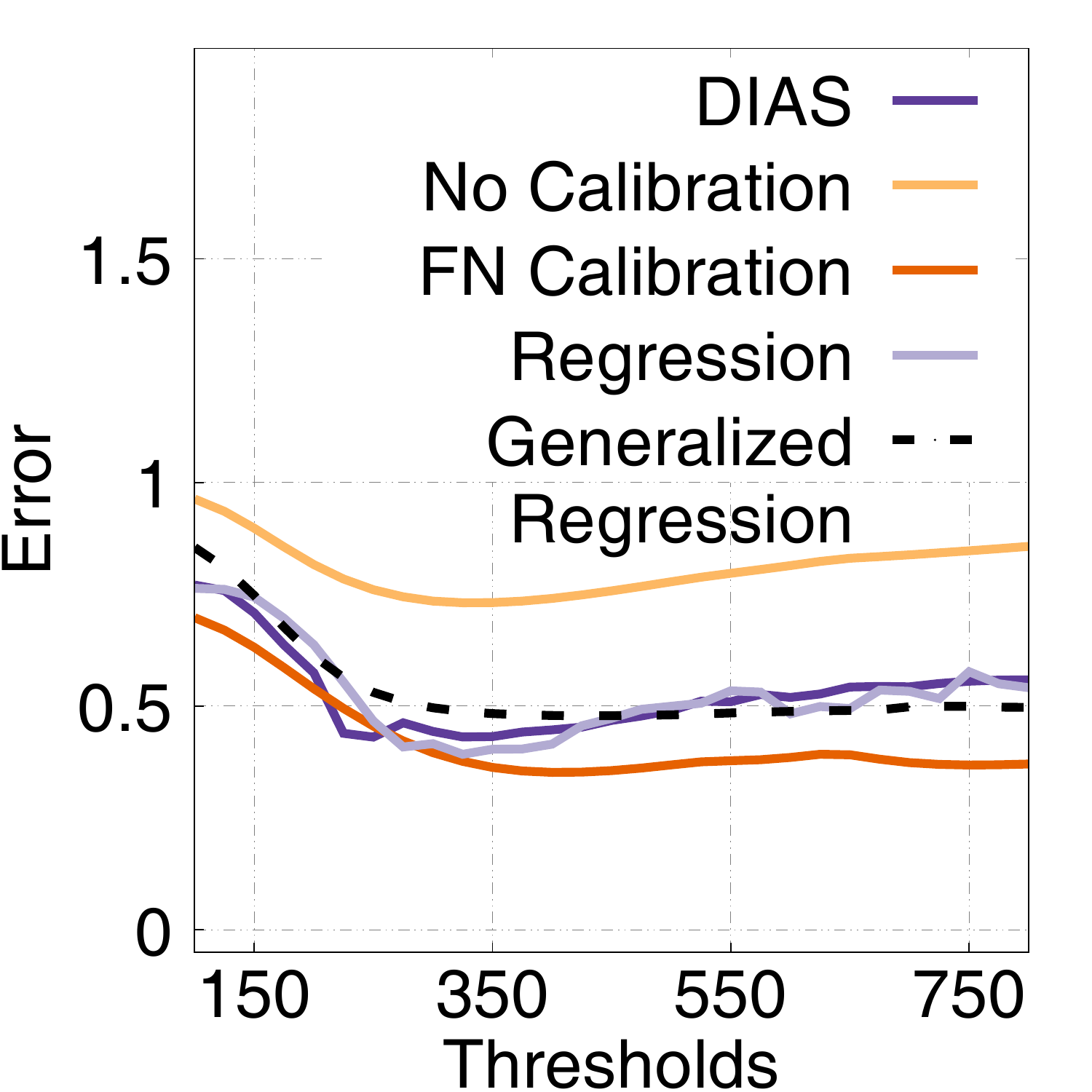}}	
\caption{Prediction performance of the calibration methods. 20\%, 50\% and 80\% fault-scale respectively for the three fault profile.}\label{fig:predictions}
\end{figure}
The followings observations are made: (i) The application agnostic calibration methods, i.e. non-calibrated prediction and false negative calibration, correlate well with DIAS, with correlation coefficient of 0.83 and 0.82, respectively across the fault profiles and scales. However, note the improvement of the latter to match the magnitude of the DIAS errors. These correlation values cannot be improved by more than  7.8\% on average with the linear regression calibration. This is an expected result as linear regression represents the best fit, i.e. linear transformation, of the application-independent features to the DIAS inconsistency cost data. On the contrary, the correlation coefficient between DIAS and generalized regression decreases to 0.79 on average. This confirms the feasibility of selecting effective thresholds for fault detection without information about the application that makes use of self-healing. (ii) Without calibration, the fault scenarios overestimate the magnitude of the DIAS errors, especially for large fault scales and thresholds. This is because the modeling of the total inconsistency cost assumes a uniform generation of inconsistency cost during runtime for each fault scenario. However, the magnitude of estimation errors in the summation aggregation function as well as how errors cancel out each other are highly dependent on the data. (iii) The hypothesis that the inconsistency cost in false negative states with the value of 1 (Table~\ref{table:scenarios}) is a worst case scenario in practice is actually confirmed: Across all thresholds, the root mean square error between DIAS and the false negative calibration is on average 0.51, 2.27, and 0.72 times lower than no calibration for a fault scale of 20\%, 50\% and 80\% respectively. These numbers are 0.44, 1.25, and 1.47 times lower across the fault scales for the \nth{1}, \nth{2} and \nth{3} fault profile respectively. 

\section{Conclusion and Future Work}\label{sec:conclusion}

This paper demonstrates how the performance of self-healing systems operating in decentralized asynchronous environments is significantly influenced by uncertainties in fault detection inherited from such systems. However, it also concludes that this influence can be accurately predicted and mitigated by modeling a number of fault scenarios that identify the origin of inconsistencies. This paper also shows how to minimize inconsistencies by tuning appropriately fault detection at the design phase. The significance of these findings stems from application-independence: A high prediction performance in the aggregation accuracy of real-world power demand data is confirmed under 696 experimental settings of different fault scales, fault profiles and fault detection thresholds. 

Future work focuses on addressing some of the limitations of this study as well as unfolding some new promising research pathways: The prediction performance of the inconsistency cost of other distributed application scenarios is required for further validation. Other costs with more complex performance trade-offs can be modeled, e.g. inconsistency vs. communication cost. Comparing the inconsistency profiles of different decentralized systems with different size, fault profiles/models, connectivity and fault-detection mechanisms can provide  further new insights on how to design, deploy and operate self-healing systems. 

\section*{Acknowledgment}

This study is supported by the Swiss National Science Foundation (SNSF) as part of the National Research Programme NRP77 Digital Transformation, project no. 187249.

\bibliography{DIAS} 
\bibliographystyle{IEEEtran}


%



\ifCLASSOPTIONcaptionsoff
  \newpage
\fi

\begin{IEEEbiography}[{\includegraphics[width=1.1in,height=1.35in,clip,keepaspectratio]{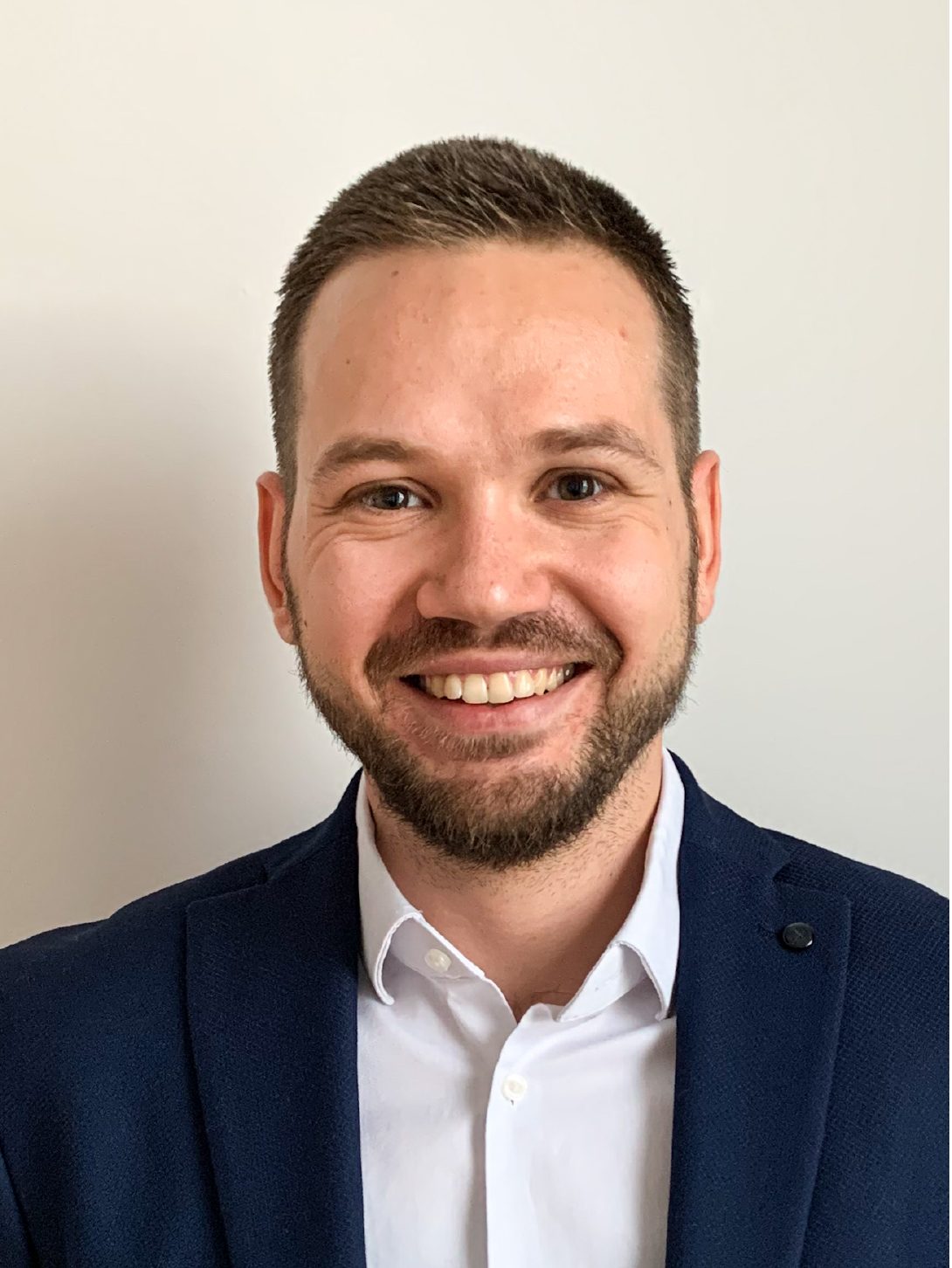}}]{Jovan Nikoli\'c}
is a Software Engineer at Google, Z\"urich, Switzerland. Since 2019, he holds a MSc in Computer Science from Swiss Federal Institute of Technology (ETH) Z\"urich, Switzerland, and since 2015 a BSc degree in Electrical Engineering and Computer Science from University of Belgrade, Belgrade, Serbia. He was a member of Computational Social Science group, ETH Z\"urich from 2015 to 2019, conducting research focusing on intelligent multi-agent systems, combinatorial multi-objective optimization and machine learning. 
\end{IEEEbiography}

\begin{IEEEbiography}[{\includegraphics[width=1.1in,height=1.35in,clip,keepaspectratio]{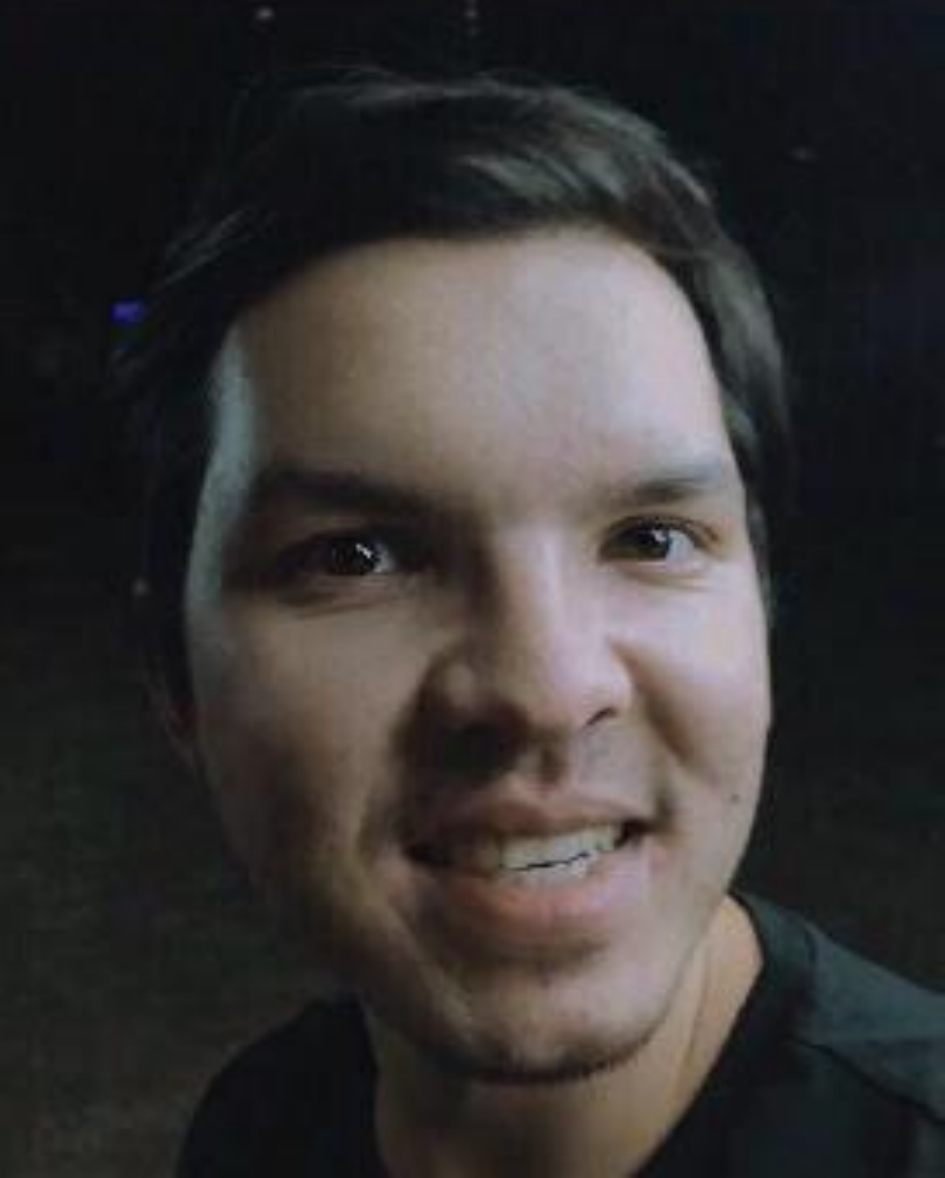}}]{Nursultan Jubatyrov}
is a Software Engineer at Facebook, London, UK. He obtained a BSc degree in Computer Science from Nazarbayev University, Nur-Sultan city, Kazakhstan in 2019. From 2017 to 2018, Nursultan was working as a Research and Engineering Assistant at Computational Social Science group, Swiss Federal Intitute of Technology (ETH)  Z\"urich, Switzerland. During this time, he conducted research on the reliability of distributed systems.
\end{IEEEbiography}


\begin{IEEEbiography}[{\includegraphics[width=1.1in,height=1.35in,clip,keepaspectratio]{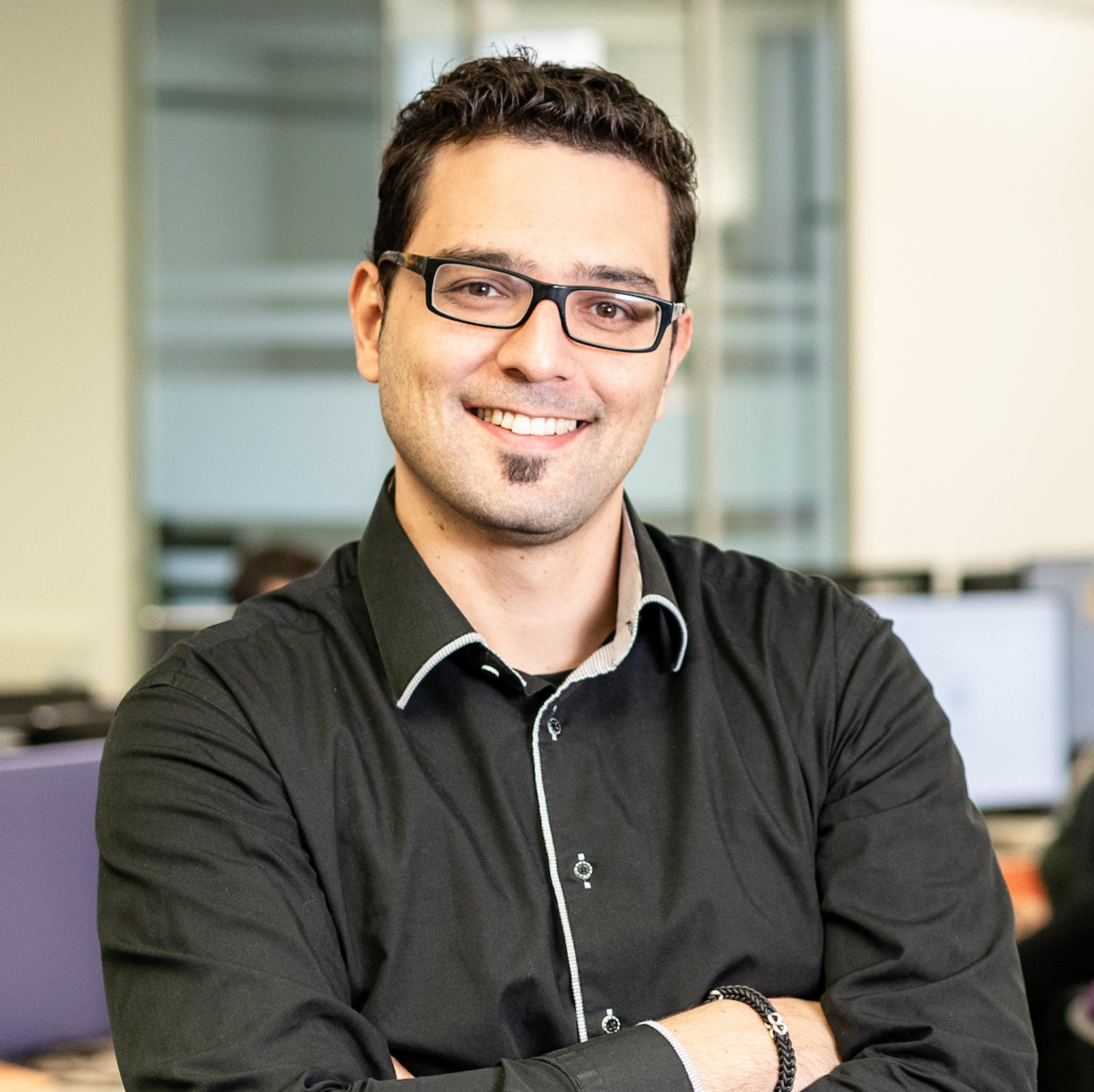}}]
{Evangelos Pournaras}
is an Associate Professor at Distributed Systems and Services group, School of Computing, University of Leeds, UK. He is also currently a research associate at UCL Center of Blockchain Technologies. He has more than 5 years experience as senior scientist and postdoctoral researcher at ETH Zurich in Switzerland after having completed his PhD studies in 2013 at Delft University of Technology in the Netherlands. Evangelos has also been a visiting researcher at EPFL in Switzerland and has industry experience at IBM T.J. Watson Research Center in the USA. Evangelos has won the Augmented Democracy Prize, the 1st prize at ETH Policy Challenge as well as 4 paper awards and honors. He has published more than 75 peer-reviewed papers in high impact journals and conferences and he is the founder of the EPOS, DIAS, SFINA and Smart Agora  projects featured at decentralized-systems.org. He has raised significant funding and has been actively involved in EU projects such as ASSET, SoBigData and FuturICT 2.0. Evangelos' research interest focus on distributed and intelligent social computing systems with expertise in the inter-disciplinary application domains of Smart Cities and Smart Grids. 
\end{IEEEbiography}

\vfill




\end{document}